\documentclass[draft,onecolumn]{IEEEtran}
\IEEEoverridecommandlockouts
\usepackage{cite}
\usepackage{amsmath,amssymb,amsfonts,amsthm}
\usepackage{algorithmic}
\usepackage{graphicx}
\usepackage{textcomp}
\usepackage{xcolor}
\usepackage{setspace} 
\usepackage{tikz}
 \usetikzlibrary{arrows.meta,positioning}
\usepackage{url}
\usetikzlibrary{positioning}
\usetikzlibrary{calc,positioning}

\usepackage{blkarray}
\usepackage{comment}

\newif\ifcomment
\commenttrue

\newcommand{\cra}[1]{\ifcomment \textcolor{blue}{Charul: #1} \fi}

\begin{document}

\newtheorem{theorem}{Theorem}
\newtheorem{definition}{Definition}
\newtheorem{corollary}{Corollary}
\newtheorem{lemma}{Lemma}
\newtheorem{example}{Example}
\newtheorem{remark}{Remark}
\newtheorem{construction}{Construction}


\title{Function-Correcting Partition Codes}

\author{\IEEEauthorblockN{ Charul Rajput\IEEEauthorrefmark{1}, B. Sundar Rajan\IEEEauthorrefmark{2}, Ragnar Freij-Hollanti\IEEEauthorrefmark{1}, Camilla Hollanti\IEEEauthorrefmark{1}}

\IEEEauthorblockA{\IEEEauthorrefmark{1}Department of Mathematics and Systems Analysis, Aalto University, Finland
    \\charul.rajput9@gmail.com, \{ragnar.freij, camilla.hollanti\}@aalto.fi}\\
    \IEEEauthorblockA{\IEEEauthorrefmark{2}Department of Electrical Communication Engineering, Indian Institute of Science, Bengaluru, India
    \\\ bsrajan@iisc.ac.in}}

\maketitle

\begin{abstract}
We introduce \emph{function-correcting partition codes (FCPCs)}, which are a natural generalization of function-correcting codes (FCCs). 
An FCPC is defined directly on a partition of the message space, rather than on a specific target function. 
We show that any $t$-error correcting code for a function $f$ is exactly an FCPC with respect to the domain partition induced by $f$, which makes these codes a natural generalization of FCCs.
We use the \emph{join} of domain partitions to construct a single code that protects multiple functions simultaneously, with application for scenarios where a receiver needs to protect more than one function value or in a broadcast scenario where the different receivers need protection for different function values. We define the notions of \emph{partition gains} to measure the bandwidth saved by using a single FCPC for multiple functions instead of constructing separate FCCs for each function. We derive general lower and upper bounds on the redundancy of such FCPCs and illustrate the achievable gains through examples. We specialize this concept of using single code for protecting multiple functions to linear functions via coset partition of the intersection of their kernels. We also present explicit FCPC constructions for locally bounded partitions and grouped weight partitions.
 Then, we associate a partition graph with any given partition of $\mathbb{F}_q^k$, and show that the existence of a suitable clique in this graph yields a set of representative information vectors that achieves the optimal redundancy. Using the existence of a full-size clique in the \emph{weight partition} and \emph{support partition}, we obtain lower and upper bounds on the optimal redundancy of FCPCs for these partitions. We introduce the notion of a \emph{block-preserving contraction} for a partition, which helps reduce the problem size of finding optimal redundancy for an FCPC. We further show that such a contraction exists for all weight-based partitions. Finally, we observe that FCPCs naturally provide a form of partial privacy in the sense that only the domain partition of the function needs to be revealed to the transmitter.

\end{abstract}
%

%
  \section{Introduction}
\label{intro}

The theory of function-correcting codes (FCCs) has so far been developed primarily for the standard one-to-one communication setting, where a single receiver is interested in obtaining the value of a particular function of the transmitted message vector. In this work, we move beyond that setting and consider a more general and practically relevant scenario. A transmitter broadcasts a message vector over a shared and error-prone link to multiple receivers, and each receiver is interested in a different function or attribute of the same message vector. Since these receivers may require protection against errors for their respective functions, the natural question is whether one can design a single code that simultaneously protects all the desired functions from up to $t$ errors.

A straightforward approach is to encode the entire message vector using a standard error-correcting code (ECC), thereby enabling complete recovery at all receivers. This trivially protects all functions but may introduce significantly more redundancy than necessary. Given that FCCs are typically studied in systematic form, another equally straightforward approach is to take the redundancies required for the different FCCs corresponding to the individual functions, append all of them to the message vector, and transmit the resulting longer vector. In this way, each receiver can simply use the redundancy meant for its own function to decode its desired value. However, this solution also incurs redundancy equal to the sum of all the individual FCC redundancies, which may be unnecessarily large. If, instead, one can find a single unified FCC that protects all the functions with strictly smaller redundancy, it leads to substantial bandwidth savings compared to both of the above approaches: (1) using an ECC that protects the entire message vector, and (2) using the sum of the redundancies of the separate FCCs designed for each function individually.

The aim is to construct codes that perform this task as required. Instead of working with the function values directly, we define codes based on partitions of the domain and refer to them as {\it function-correcting partition codes (FCPCs)}. The setup is illustrated in Fig. \ref{fig:FCPC_setup}. Since every function naturally partitions its domain into preimage sets, these FCPCs can be used as FCCs. In fact, each function induces a natural partition of the domain into its preimage sets, and an FCC is exactly the FCPC defined with respect to this induced partition. This makes FCPCs a direct generalization of FCCs.


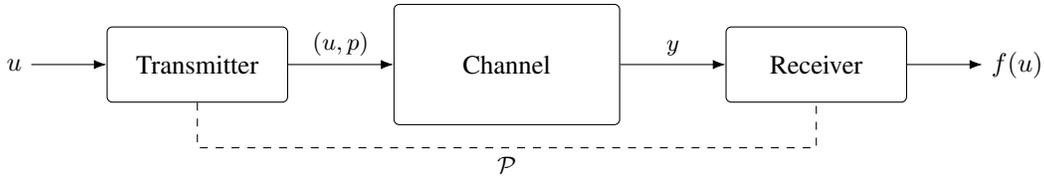
\begin{figure}[t]
\centering
\begin{tikzpicture}[
    >=Latex,
    node distance=18mm and 18mm,
    box/.style={draw, rounded corners=2pt, minimum height=10mm, minimum width=24mm, align=center},
    bigbox/.style={draw, rounded corners=2pt, minimum height=16mm, minimum width=30mm, align=center},
    lab/.style={font=\small}
]
\node (u) {$u$};
\node[box, right=10mm of u] (tx) {Transmitter};
\node[bigbox, right=14mm of tx] (ch) {Channel};
\node[box, right=14mm of ch] (rx) {Receiver};
\node[right=10mm of rx] (fu) {$f(u)$};

\draw[->] (u) -- node[above,lab] { } (tx.west);

\draw[->] (tx.east) -- node[above,lab] {$(u,p)$} (ch.west);
\draw[->] (ch.east) -- node[above,lab] {$y$} (rx.west);
\draw[->] (rx.east) -- (fu);

\draw[-, dashed]
  (tx.south) -- ++(0,-6mm) --
  node[below,lab] {$\mathcal{P}$}
  ++($(rx.south)-(tx.south)$) --
  (rx.south);

\end{tikzpicture}

\caption{Illustration of the function-correcting partition code (FCPC) setup.
A message $u$ is transmitted with redundancy $p$ over a noisy channel.
The transmitter has access only to a partition of the message domain and does not know the complete function desired by the receiver.
Using the received word $y$ and the knowledge of the function $f$, the receiver correctly recovers the function value $f(u)$.}
\label{fig:FCPC_setup}
\end{figure}

This line of study is also directly useful in scenarios where there is only a single receiver but the receiver is interested in protecting the values of multiple functions of the same message vector. Even though the communication is one-to-one, the requirement of safeguarding several functions simultaneously leads to the same core problem: either one protects the entire message vector, or one appends the individual FCC redundancies for each function. As before, both approaches may introduce more redundancy than necessary. Hence, designing a unified FCPC that jointly protects all the required functions can yield significant savings. The framework developed in this work therefore applies both to multiple-receiver settings and to single-receiver, multi-function settings, as illustrated in Fig.~\ref{fig:MFS}.

Further, we refer to the savings, arising from using a single code for multiple functions instead of using separate FCCs for each function, as the {\it partition redundancy gain} and {\it relative rate increment}.

There is another benefit of using FCPCs in a network setting, as they
provide a degree of partial privacy regarding the function being computed,
which we refer to as \emph{function class privacy}. Here, a function class
consists of all functions on the domain that induce the same domain partition.
In FCCs, both the receiver and the transmitter must know the exact function. In
contrast, an FCPC requires the transmitter to know only the domain
partition induced by the function that the receiver wishes to compute. Since
many different functions can belong to the same function class and induce the
same domain partition, the transmitter cannot infer the exact function from the partition alone. Further discussion of this notion is given in
Section~\ref{sec:privacy}.

\begin{figure}[ht]
\centering
\begin{tikzpicture}[>=stealth]


\node[draw, rectangle] (S1) at (0,0) {Transmitter};
\coordinate (B1) at (0,-1.2);        

\draw[very thick] (S1) -- (B1);

\node[draw, rectangle] (R1) at (-2,-3) {$R_1$};
\node[draw, rectangle] (R2) at (0,-3) {$R_2$};
\node at (1.2,-3) (dots) {$\cdots$};
\node[draw, rectangle] (RK) at (2.4,-3) {$R_K$};

\draw (B1) -- (R1);
\draw (B1) -- (R2);
\draw (B1) -- (RK);

\node at (-2,-4.0) {$f^{(1)}$};
\node at ( 0,-4.0) {$f^{(2)}$};
\node at ( 2.4,-4.0) {$f^{(K)}$};

\node at (-4.0,-3.0) {\textbf{Receivers}};
\node at (-4.0,-4.0) {\textbf{Functions}};

\node at (0,-5.2) {\small Fig.~2.1: Multiple receivers.};
\end{tikzpicture}
\hspace{2cm}
\begin{tikzpicture}[>=stealth]


\node[draw, rectangle] (S2) at (0,0) {Transmitter};
\node[draw, rectangle] (R)  at (0,-3) {Receiver};

\draw (S2) -- (R);

\node at (0,-4.0) {$f^{(1)},\, f^{(2)},\, \ldots,\, f^{(K)}$};

\node at (0,-5.2) {\small Fig.~2.2: Single receiver multi-function.};

\end{tikzpicture}

\caption{Multiple receivers and single–receiver multi-function settings.}
\label{fig:MFS}
\end{figure}

\subsection{Related works}
The notion of FCCs was formally introduced in \cite{LBWY2023}, which laid the foundational framework for studying error correction when the objective is to reliably recover a function of the transmitted message rather than the message itself. In this work, the authors defined the key parameters of FCCs, established fundamental redundancy bounds, and investigated structural properties of such codes. They also provided explicit constructions for several important classes of functions, including locally binary functions, the Hamming weight function, and the Hamming weight distribution function. For some of these cases, the proposed constructions were shown to be optimal by achieving the corresponding lower bounds on redundancy. The initial development focused on communication over binary symmetric channels.

Subsequent research has extended this framework to more general channel models. In \cite{XLC2024}, FCCs were studied in the context of symbol-pair read channels over binary fields, and this line of work was further generalized to $b$-symbol read channels in \cite{SSY2025}. 
In \cite{PR2025}, the authors leveraged the fact that the kernel of a linear function forms a subspace of $\mathbb{F}_q^k$, and used the associated coset partition to construct FCCs for linear functions. They also derived a general lower bound on the redundancy of FCCs, which specializes to a classical lower bound for ECCs in the case of bijective functions, since FCCs reduce to ECCs in this setting. The tightness of this bound was established for certain parameters. The work \cite{GXZZ2025} further refined the study of FCCs for the Hamming weight and Hamming weight distribution functions, providing improved bounds and constructions tailored to these specific functions.

In a related direction, broader classes of functions have been considered. In \cite{RRHH2025ITW25}, the notion of locally binary functions was generalized to locally $(\rho,\lambda)$-bounded functions, and an upper bound on the redundancy of FCCs for this class was derived. Moreover, the authors showed that any function on $\mathbb{F}_2^k$ can be viewed as locally $(\rho,\lambda)$-bounded for appropriate choices of parameters, thereby extending the applicability of their results to a general class of functions. An extension of this framework to the $b$-symbol read channel was subsequently presented in \cite{VSS2025}. Additionally, \cite{LS2025} established an upper bound on FCC redundancy that is within a logarithmic factor of the lower bound in \cite{LBWY2023}, and proved that this lower bound is tight for sufficiently large field sizes.

Another line of work studies FCCs under alternative distance measures and channel models. In \cite{VS2025,HUR2025}, FCCs are investigated in the Lee metric, where Plotkin-type bounds are derived for general functions and explicit constructions are provided for several function classes, with optimality established in certain cases. The work \cite{LL2025} considers FCCs with homogeneous distance, developing redundancy bounds and constructions adapted to this setting. In addition, Plotkin-type bounds for FCCs over $b$-symbol read channels are studied in \cite{SR2025}.

More recently, a related direction was explored in \cite{RRHH2025}, where the authors study FCCs with data protection, aiming to protect both the computed function value and the underlying data. The work introduces a general framework for this joint protection requirement and derives corresponding redundancy bounds, adding another variant to the study of FCCs.

\subsection{Contributions and organization of the paper}

The main contributions of this work are summarized below.
\begin{itemize}
    \item We introduce $(\mathcal{P},t)$-encodings, called function-correcting partition codes, defined directly on a partition $\mathcal{P}$ of $\mathbb{F}_q^k$. We show that any FCC is exactly an FCPC with respect to the domain partition induced by $f$, which makes FCPCs a natural generalization of FCCs.
    
    \item We use the join of domain partitions to construct a single code that protects multiple functions simultaneously, and we specialize this to linear functions via coset partition of the intersection of their kernels.
    
 \item We define the notions of \emph{partition redundancy gain} and \emph{relative rate increment} to measure the bandwidth saved by using a single FCPC for multiple functions instead of constructing separate FCCs for each function. We derive general lower and upper bounds on the redundancy of such FCPCs and illustrate the achievable gains through examples.

\item We present an explicit optimal FCPC construction for \emph{locally bounded partitions}, obtained as a reformulation of the FCC construction for locally bounded functions introduced in \cite{LBWY2023}. We also propose FCPC constructions for \emph{grouped weight partitions} by employing Gray codes and extending the approach of \cite{GXZZ2025}. Furthermore, for the subclass of \emph{consecutive grouped weight partitions}, we identify conditions under which the proposed construction achieves optimal redundancy.

\item We associate a partition graph with any given partition of $\mathbb{F}_q^k$, and show that the existence of an appropriately chosen clique in this graph allows the problem of determining the optimal redundancy to be reduced to analyzing only the vectors corresponding to the vertices of that clique. Then we show the existence of such cliques for three 
classes of partitions: coset partitions, Hamming-weight partitions, and support partitions.

\item Using the existence of a full-size clique in the weight partition and the support partition of $\mathbb{F}_q^k$, we obtain lower and upper bounds on the optimal redundancy of FCPCs for these partitions.

\item We introduce the notion of a \emph{block-preserving contraction} for a partition, which helps reduce the problem size of finding optimal redundancy for an FCPC. We further show that such a contraction exists for all weight-based partitions. We also give a condition for cosets under which the partition graph of the corresponding coset partition admits a block-preserving contraction.

\item FCPCs inherently provide \emph{function class privacy} by revealing only the domain partition of the function to the transmitter. This leads to ambiguity among functions within the same function class, without affecting the existence or characterization of optimal FCPCs.
\end{itemize}

The rest of the paper is organized as follows. Section \ref{preliminaries} recalls the basic notions on partitions, graphs, and FCCs. In Section \ref{par_code}, we introduce FCPCs, study their relation to FCCs, consider their use for protecting multiple functions, and define the partition gains together with general redundancy bounds, including the case of linear functions. Section~\ref{par_graph} introduces the partition graph and then focuses on coset partitions, weight-based partitions and support partition. Section~\ref{sec:BPC} presents the notion of block-preserving contraction of partitions and establishes its connection to the optimal redundancy of the corresponding FCPC. Section \ref{sec:privacy} discusses the notion of structural privacy inherent to FCPCs and illustrates the resulting ambiguity
among functions inducing the same domain partition. Finally, Section~\ref{conclusion} concludes the paper.

\subsection{Notations}

Throughout the paper, $\mathbb{F}_q$ denotes the finite field with $q$ elements, and $\mathbb{F}_q^k$ denotes the $k$-dimensional vector space over it. For vectors $u,v \in \mathbb{F}_q^k$, $\mathrm{wt}(u)$ and $d(u,v)$ denote Hamming weight and Hamming distance, respectively. For a function $f$, notation $\mathrm{Im}(f)$ denotes its image. For a linear function $f:\mathbb{F}_q^k \to \mathbb{F}_q^{\ell}$, its kernel is denoted by $\ker(f)$. We use $[n]$ to denote the index set $\{1,2,\ldots,n\}$. A subspace of $\mathbb{F}_q^k$ is denoted by $U \le \mathbb{F}_q^k$. For a vector $u=(u_1,\ldots,u_k)\in\mathbb{F}_q^k$, its support is defined as $\mathrm{supp}(u) = \{\, i\in[k] : u_i \ne 0 \,\}.$ For two sets $A$ and $B$, their symmetric difference is $A \Delta B = (A\setminus B) \cup (B \setminus A).$
In several examples, we write vectors such as $(0,1,1)$ in concatenated form as $011$.


\section{Preliminaries}
\label{preliminaries}
This section contains some basic definitions related to partitions \cite{DF2004, LN}, graphs \cite{BM2008} and FCCs \cite{LBWY2023}.

\begin{definition}[Partition]
A partition $\mathcal{P}$ of a finite set $S$ is a set of pairwise disjoint subsets of $S$ whose union is $S$. The elements of $\mathcal{P}$ are called its blocks.
\end{definition}

\begin{definition}[Refinement of a partition]
Let $\mathcal{P}$ be a partition of a set $S$. Another partition $\mathcal{Q}$ of the set $S$ is called a refinement of $P$ if for each block $Q$ in $\mathcal{Q}$, there exists a block $P$ in $\mathcal{P}$ such that $Q \subseteq P$. We say that the partition $\mathcal{Q}$  is finer than $\mathcal{P}$ or, equivalently, $\mathcal{P}$ is coarser than $\mathcal{Q}$.
\end{definition}

%

\begin{definition}[Coarsest common refinement or join]
Given two partitions $\mathcal{P}$ and $\mathcal{Q}$ of a set $S$, their join, denoted by $\mathcal{P} \vee \mathcal{Q}$, is the coarsest partition that is a refinement of both $\mathcal{P}$ and $\mathcal{Q}$. This means that each block in the join is the intersection of a block from $\mathcal{P}$ and a block from $\mathcal{Q}$.
\end{definition}

\begin{example}
Let $S=\{1,2,3,4,5,6,7\}$, and $\mathcal{P}=\{\{1,3,4\},\{2,5\},\{6,7\}\}$ and $\mathcal{Q}=\{\{1,2,5\},\{3,4\},\{6\},$ $\{7\}\}$ be two partitions of $S$. Then the coarsest common refinement of $\mathcal{P}$ and $\mathcal{Q}$ is
$$\mathcal{P} \vee \mathcal{Q}=\{\{1\},\{2,5\},\{3,4\},\{6\},\{7\}\}.$$
\end{example}

\begin{definition}[Graph]
A \emph{graph} is an ordered pair $G=(V,E)$, where $V$ is a finite nonempty set of vertices, and 
$
E \subseteq \{\,(u,v) : u,v \in V,\; u \neq v \,\}
$
is a set of edges. 
\end{definition}

Two distinct vertices $u,v \in V$ are said to be \emph{adjacent} in a graph $G=(V,E)$ if $(u,v) \in E$, and $G$ is called \emph{complete} if every pair of distinct vertices is adjacent.  

\begin{definition}[Clique]
A \emph{clique} in a graph $G=(V,E)$ is a set of vertices $C \subseteq V$ such that every two distinct vertices in $C$ are adjacent, i.e.,
$
(u,v) \in E$ for all distinct $u,v \in C.
$
\end{definition}

\begin{definition}[$t$-partite Graph]
A graph $G=(V,E)$ is said to be \emph{$t$-partite} if its vertex set admits a partition
\[
V = V_1 \cup V_2 \cup \cdots \cup V_t,
\qquad 
V_i \neq \emptyset,\;\; V_i \cap V_j = \emptyset\ (i \neq j),
\]
such that no two vertices within the same part are adjacent, i.e.,
$(u,v) \notin E$ whenever $u,v \in V_i.
$
\end{definition}

\begin{definition}[Function-correcting codes]
Consider a function $f: \mathbb{F}_q^k \rightarrow S$. A systematic encoding $\mathcal{C}: \mathbb{F}_q^k \rightarrow \mathbb{F}_q^{k+r}$ is defined as an $(f, t)$-function correcting code (FCC) if, for any $u_1, u_2 \in \mathbb{F}_q^k$ such that $f(u_1) \neq f(u_2)$, the following condition holds: 
$$d(\mathcal{C}(u_1), \mathcal{C}(u_2)) \geq 2t+1,$$
where $d(x, y)$ denotes the Hamming distance between vectors $x$ and $y$.
\end{definition}

The \emph{optimal redundancy} $r_f(k, t)$ is defined as the minimum of $r$ for which there exists an $(f, t)$-FCC with an encoding function $\mathcal{C}: \mathbb{F}_2^k \rightarrow \mathbb{F}_2^{k+r}$.

\begin{definition}[Distance Requirement Matrix (DRM)]
Let $u_1, u_2,$ $\ldots, u_{M} \in \mathbb{F}_q^k$. The distance requirement matrix (DRM) $\mathcal{D}_f(t, u_1, u_2,\ldots, u_M)$ for an $(f, t)$-FCC is a $M \times M$ matrix with entries
$$[\mathcal{D}_f(t, u_1, \ldots, u_M)]_{i, j} = \begin{cases}  \max(2t+1-d(u_i, u_j), 0), & \text{if} \ f(u_i) \neq f(u_j), \\
0 & \text{otherwise},
\end{cases}$$
where $i, j \in \{1,2, \ldots, M\}$.
\end{definition}

When $M=2^k,$ we get the DRM for $\mathcal{C}:\mathbb{F}_q^k \rightarrow \mathbb{F}^{k+r}.$

\begin{definition}[Irregular-distance code or $\mathcal{D}$-code]
Let $\mathcal{D} \in \mathbb{N}^{M\times M}$. Then $\mathcal{P}=\{p_1, p_2, \ldots, p_M\}$, where $p_i \in \mathbb{F}_{q}^r, ~i=1,2,\ldots, M$, is said to be an irregular-distance code or $\mathcal{D}$-code if there is an ordering of $P$ such that $d(p_i, p_j) \geq [\mathcal{D}]_{i,j}$ for all $i, j \in \{1, 2, \ldots, M\}$. Further, $N_q(\mathcal{D})$ is defined as the smallest integer $r$ such that there exists a $\mathcal{D}$-code of length $r$. If $[\mathcal{D}]_{i, j} = D$ for all $i, j \in \{1,2,\ldots, M\}, i\neq j$, then $N_q(\mathcal{D})$ is denoted as $N_q(M, D)$.
\end{definition}


For a function $f: \mathbb{F}_q^k \mapsto Im(f)$, the distance between $f_i, f_j \in Im(f)$ is defined as
$$d(f_i, f_j) = \min_{u_1, u_2 \in \mathbb{F}_q^k} \{d(u_1, u_2) |  f(u_1)=f_i, f(u_2) = f_j\}.$$

\begin{definition}[Function Distance Matrix (FDM)]
Consider a function $f: \mathbb{F}_q^k \mapsto Im(f)$ and $E=|Im(f)|$. Then the  $E \times E$ matrix $\mathcal{D}_f(t, f_1, f_2,\ldots, f_E)$ with entries given as
$$[\mathcal{D}_f(t, f_1, f_2,\ldots, f_E)]_{i, j} = \begin{cases}  \max(2t+1 -d(f_i, f_j), 0), & \text{if} \ i \neq j, \\
0 & \text{otherwise},
\end{cases}$$
is called the function distance matrix.
\end{definition}

\begin{corollary}[{\cite[Corollary 1]{LBWY2023}}]\label{thm1}
For any function $f: \mathbb{F}_q^k \mapsto S$ and $\{u_1, u_2, \ldots, u_m\}\subseteq \mathbb{F}_q^k$,
$$r_f(k, t) \geq N_q(\mathcal{D}_f(t, u_1, u_2, \ldots, u_m)),$$
and for $|Im(f)|\geq 2$, $r_f (k, t) \geq 2t$.
\end{corollary}

\begin{theorem}[{\cite[Theorem 2]{LBWY2023}}]\label{thm2}
For any function $f: \mathbb{F}_q^k \mapsto Im(f)=\{f_1, f_2, \ldots, f_E\}$,
$$r_f(k, t) \leq N_q(\mathcal{D}_f(t, f_1, f_2, \ldots, f_E)),$$
where $\mathcal{D}_f(t, f_1, f_2, \ldots, f_E)$ is a FDM.
\end{theorem}

\begin{corollary}[{\cite[Corollary 2]{LBWY2023}}]\label{col1}
If there exists a set of representative information vectors $u_1, u_2, \ldots, u_{E}$ with $\{f(u_1), f(u_2),$ $\ldots,$ $f(u_E)\} = Im(f)$ and $\mathcal{D}_f(t, u_1, u_2, \ldots, u_E)=\mathcal{D}_f(t, f_1, f_2, \ldots, f_E)$, then
$$r_f(k, t) = N_q(\mathcal{D}_f(t, f_1, f_2, \ldots, f_E)).$$
\end{corollary}

A generalization of the Plotkin bound over $\mathbb{F}_q$ is given in \cite{RRHH2025} as follows.
\begin{lemma}[{\cite[Lemma 13]{RRHH2025}}]\label{lem:RRHH}
	For any distance matrix $D\in \mathbb{N}^{M\times M}$ and for irregular
	distance codes over $\mathbb{F}_q$, we have
	$$
	N_q(D) \ge
	\frac{2q}{M^2(q-1)-a(q-a)}
	\sum_{1 \le i < j \le M} [D]_{i,j},
	$$
	where $a = M \bmod q$.
\end{lemma}


\section{Function-correcting partition codes}\label{par_code}

In this section, we first define the $(\mathcal{P},t)$-encoding for an arbitrary partition $\mathcal{P}$ of $\mathbb{F}_q^k$, and then illustrate how this framework applies to multiple functions and linear functions.

\begin{definition}[Function domain partition]
Consider a function $f: \mathbb{F}_q^k \rightarrow S$ such that $\mathrm{Im}(f)=\{f_1, f_2, \ldots,$ $f_E\}$, where $|\mathrm{Im}(f)|=E$. A partition $\mathcal{P}_f=\{P_1, P_2, \ldots, P_E\}$ of $\mathbb{F}_q^k$ defined as
 $$P_i=\{u\in \mathbb{F}_q^k \mid f(u)=f_i\},$$
is called the domain partition of $f$. In other words, $\mathcal{P}_f=\{f^{-1}(f_i) \mid i \in [E]\}$.
\end{definition}


The properties of an $(f,t)$-FCC depend only on the partition of $\mathbb{F}_q^k$ induced by $f$, and not on the specific values of the function.  
Hence, different functions with the same domain partition admit the same FCC.
Motivated by this, we introduce a more general notion of encoding based purely
on a partition of $\mathbb{F}_q^k$.

\begin{definition}[Function-correcting partition codes]
Let $\mathcal{P}=\{P_1,P_2,\ldots,P_E\}$ be a partition of $\mathbb{F}_q^k$.
A systematic mapping 
$\mathcal{C}_{\mathcal{P}} : \mathbb{F}_q^k \rightarrow \mathbb{F}_q^{k+r}$
is called a \emph{$(\mathcal{P},t)$-encoding} if for all $u\in P_i$ and
$v\in P_j$ with $i\neq j$,
$$
d\big(\mathcal{C}_{\mathcal{P}}(u), \mathcal{C}_{\mathcal{P}}(v)\big)
\ge 2t+1.
$$
\end{definition}

\begin{definition}[Optimal redundancy]
The \emph{optimal redundancy} $r_{\mathcal{P}}(k, t)$ is defined as the minimum of $r$ for which there exists a $(\mathcal{P}, t)$-encoding $\mathcal{C}: \mathbb{F}_2^k \rightarrow \mathbb{F}_2^{k+r}$.
\end{definition}

The notions of DRM and $\mathcal{D}$-code extend naturally from functions to arbitrary partitions.

\begin{definition}[Partition Distance Requirement Matrix (PDRM)]
Let $\mathcal{P} = \{P_1,P_2,\ldots,P_E\}$ be a partition of $\mathbb{F}_q^k$ and let $u_1,u_2,\ldots,u_M \in \mathbb{F}_q^k$. The partition distance requirement matrix
$\mathcal{D}_{\mathcal{P}}(t,u_1,\ldots,u_M)$ for a $(\mathcal{P},t)$-encoding is a $M \times M$ matrix with entries
\[
[\mathcal{D}_{\mathcal{P}}(t,u_1,\ldots,u_M)]_{i,j}
=
\begin{cases}
\max(2t+1 - d(u_i,u_j), 0), & \text{if } u_i \in P_a,\, u_j \in P_b \text{ for some } a \neq b,\\[2mm]
0, & \text{otherwise},
\end{cases}
\]
for all $i,j \in \{1,2,\ldots,M\}$.
\end{definition}

For $\mathcal{D}=\mathcal{D}_\mathcal{P}(t, u_1, u_2,\ldots, u_{q^k})$, if we have a $\mathcal{D}$-code $Z=\{z_1, z_2, \ldots, z_{q^k}\}$, then it can be used to construct a $(\mathcal{P}, t)$-encoding $\mathcal{C}(u_i) = (u_i, z_i)$ for all $i \in \{1, 2, \ldots, q^k\}$.

\begin{definition}[Block distance]
Let $\mathcal{P}=\{P_1,P_2,\ldots,P_E\}$ be a partition of $\mathbb{F}_q^k$. For two distinct blocks $P_i$ and $P_j$, their distance is defined as
$$  d(P_i,P_j)= \min_{u \in P_i,\; v \in P_j} d(u,v).
$$
For $i=j$, we set $d(P_i,P_j)=0$.
\end{definition}

\begin{definition}[Partition Distance Matrix (PDM)]
Let $\mathcal{P}=\{P_1,P_2,\ldots,P_E\}$ be a partition of $\mathbb{F}_q^k$. The $E\times E$ matrix
$   \mathcal{D}_{\mathcal{P}}(t;P_1,\ldots,P_E)$
with entries
$$
    [\mathcal{D}_{\mathcal{P}}(t;P_1,\ldots,P_E)]_{i,j}
    =
    \begin{cases}
        \max\big( 2t+1 - d(P_i,P_j),\, 0 \big), & i\neq j, \\
        0, & i=j,
    \end{cases}
$$
is called the \emph{partition distance matrix} (PDM).
\end{definition}
Clearly, for a function $f:\mathbb{F}_q^k \to \mathrm{Im}(f)$ with domain partition $\mathcal{P}_f$, we have $\mathcal{D}_{f}(t, u_1,\ldots,u_{q^k})$ $= \mathcal{D}_{\mathcal{P}_f}$ $(t;P_1,\ldots,P_E)$.

The following bounds are immediate extensions of the corresponding results 
(Theorem~\ref{thm2}, Corollary~\ref{thm1}, and Corollary~\ref{col1}) from \cite{LBWY2023}. 
Since the proofs in \cite{LBWY2023} rely only on the induced function domain partition 
(and not on the specific values of the function), 
the same reasoning follows identically in the partition setting, and we state the results directly.

\begin{theorem}\label{thm:LUB_partition}
Let $\mathcal{P}=\{P_1,P_2,\ldots,P_E\}$ be a partition of $\mathbb{F}_q^k$ and 
let $\{u_1,\ldots,u_m\}\subseteq\mathbb{F}_q^k$. Then
$$
    N\left(\mathcal{D}_{\mathcal{P}}(t,u_1,\ldots,u_m)
    \right) \le r_{\mathcal{P}}(k,t)
    \le
    N\left(\mathcal{D}_{\mathcal{P}}(t;P_1,\ldots,P_E)
    \right).
$$
In particular, if $|\mathcal{P}|\ge 2$, then $r_{\mathcal{P}}(k,t)\ge 2t$.
\end{theorem}

\begin{corollary}\label{col1_p}
Suppose, there exist representatives 
$u_i \in P_i$ for $i\in [E]$ such that $\mathcal{D}_{\mathcal{P}}(t,u_1,\ldots,u_E) = \mathcal{D}_{\mathcal{P}}(t;P_1,\ldots,P_E).$
Then
$$r_{\mathcal{P}}(k,t) = N\left(
        \mathcal{D}_{\mathcal{P}}(t,u_1,\ldots,u_E)
    \right).
$$
\end{corollary}

Next, we examine how the relationship between two partitions, with one being a refinement of the other, affects their FCPCs' redundancy.

\begin{lemma}\label{lem1}
    Let $\mathcal{P}$ and $\mathcal{Q}$ be two partitions of $\mathbb{F}_q^k$. If $\mathcal{Q}$ is a refinement of $\mathcal
    P$, then any $(\mathcal{Q}, t)$-encoding is also a $(\mathcal{P}, t)$-encoding, and
    $$r_{\mathcal{P}}(k, t) \leq r_{\mathcal Q}(k,t).$$
\end{lemma}

\begin{proof}
   Let $\mathcal{Q}$ be a refinement of $\mathcal{P}$, and let $\mathcal{C}:\mathbb{F}_q^k \to \mathbb{F}_q^{k+r}$ be a $(\mathcal{Q}, t)$-encoding. 
   Let $u,v \in\mathbb{F}_q^k$ such that $u \in P_1$ and $v\in P_2$ for some distinct $P_1,P_2 \in \mathcal{P}$.  Since $\mathcal{Q}$ is a refinement of $\mathcal{P}$, there exists $Q_1,Q_2 \in \mathcal{Q}, Q_1 \neq Q_2$ such that $u\in Q_1 \subseteq P_1$ and $v\in Q_2 \subseteq P_2$. Now, by the definition of a $(\mathcal{Q}, t)$-encoding $\mathcal{C}$,  we have 
$$d(\mathcal{C}(u), \mathcal{C}(v)) \geq 2t+1.$$
Therefore, $\mathcal{C}$ is also a $(\mathcal{P}, t)$-encoding, and the optimal redundancy of it is bounded as $r_{\mathcal{P}}(k, t) \leq r_{\mathcal Q}(k,t).$
\end{proof}

An $(n, q^k,2t+1)$ ECC is a $(\mathcal{Q},t)$-encoding for the finest partition $\mathcal{Q}=\big\{\{x\} : x \in \mathbb{F}_q^k\big\}$. Since this finest partition is a refinement of any partition of $\mathbb{F}_q^k$, such a code is a $(\mathcal{P},t)$-encoding for every partition $\mathcal{P}$.

\begin{lemma}  \label{lem2}
A $(\mathcal{P}, t)$-encoding defined above for a partition $\mathcal{P}$ of $\mathbb{F}_q^k$,  is an $(f, t)$-FCC for any function $f: \mathbb{F}_q^k \rightarrow S$ with domain partition $\mathcal{P}_f=\mathcal{P}$. 
\end{lemma}

\begin{proof}
 Consider a function $f: \mathbb{F}_q^k \rightarrow S$ with $\mathcal{P}_f=\{P_1, P_2, \ldots, P_E\}$, and a $(\mathcal{P}_f, t)$-encoding $\mathcal{C}_{\mathcal{P}_f}$. Let $u, v \in \mathbb{F}_q^k$ be such that $f(u) \neq f(v)$. Then  $u\in P_i$ and $v\in P_j$ for some $i,j \in [E], i\neq j$, therefore, by the definition of $\mathcal{C}_{\mathcal{P}_f}$,  we have 
$d(\mathcal{C}_{\mathcal{P}_f}(u), \mathcal{C}_{\mathcal{P}_f}(v)) \geq 2t+1.$
\end{proof} 

Thus, for every function $f$, we have
$r_f(k,t) = r_{\mathcal{P}_f}(k,t).$
The following lemma shows how redundancy behaves under functional composition.
\begin{lemma}\label{lem:composition}
Let $h:\mathbb{F}_q^k \to S$ be a function, and let 
$g:\mathrm{Im}(h) \to S'$ be a function defined on the image of $h$.  
Consider $f:\mathbb{F}_q^k \to S'$ defined by
$f(x) = g(h(x))$ for all $x \in \mathbb{F}_q^k.$
Then, for every integer $t \ge 0$, we have
$$ r_f(k,t) \le r_h(k,t). $$
Moreover, if $g$ is a bijection on $\mathrm{Im}(h)$, then
$
    r_f(k,t) = r_h(k,t).
$
\end{lemma}
\begin{proof}
Let $\mathcal{P}_h$ and $\mathcal{P}_f$ denote the domain partitions of $h$ and $f$, respectively.  If $x, y\in\mathbb{F}_q^k$ are in the same block of $\mathcal{P}_h$, then $h(x)=h(y)$, so $f(x)=g(h(x))=g(h(y))=f(y)$, which means $x$ and $y$ are also in the same block of $\mathcal{P}_f$. Therefore, $\mathcal{P}_h$ is a refinement of $\mathcal{P}_f$, so by Lemma~\ref{lem1} we get
$
r_{\mathcal{P}_f}(k,t) \le r_{\mathcal{P}_h}(k,t),
$
and equivalently,
$r_f(k,t) \le  r_h(k,t) $ as claimed.
Furthermore, if $g$ is bijective on $\mathrm{Im}(h)$ then we can also write $h=g^{-1}\circ f$, so the partitions $\mathcal{P}_h$ and $\mathcal{P}_f$ are refinements of each other, and therefore equal as partitions. It follows that $
r_f(k,t) = r_{\mathcal{P}_f}(k,t)
= r_{\mathcal{P}_h}(k,t)
= r_h(k,t).
$
\end{proof}

\subsection{FCPCs for protecting multiple functions}

There can be multiple functions that share the same domain partition.
Therefore, one FCPC can be used for multiple functions simultaneously. If $f^{(1)}, f^{(2)}, \ldots, f^{(K)}$ are the functions on $\mathbb{F}_q^k$ such that $\mathcal{P}_{f^{(1)}}=\mathcal{P}_{f^{(2)}}=\cdots=\mathcal{P}_{f^{(K)}}=\mathcal{P}$, then the encoding $\mathcal{C_{\mathcal{P}}}$ is an $(f^{(i)}, t)$-FCC for all $i\in [K]$. The following example illustrates this.

\begin{example}\label{ex2}
Consider the following three functions from $\mathbb{F}_2^4$ to $\mathbb{F}_2^2$ defined as
\begin{align*}
f^{(1)}: \qquad &x \mapsto \begin{bmatrix}  1 & 1& 1& 0 \\ 0 & 1 & 1 & 0  \end{bmatrix} x \\
f^{(2)}: \qquad &x \mapsto \begin{bmatrix}  1 & 1& 1& 0 \\ 1 & 0 & 0 & 0  \end{bmatrix} x \\
f^{(3)}: \qquad &x \mapsto \begin{bmatrix}  1 & 0& 0& 0 \\ 0 & 1 & 1 & 0  \end{bmatrix} x 
\end{align*}
All these functions have the same function domain partition, i.e., $\mathcal{P}_{f^{(1)}}=\mathcal{P}_{f^{(2)}}=\mathcal{P}_{f^{(3)}}=\mathcal{P}$, where
\begin{align*}
\mathcal{P}=\{&P_1=\{ 0000, 0001, 0110, 0111 \},
                    P_2=\{ 0010, 0011, 0100, 0101 \}, \\
                     &P_3=\{ 1000, 1001, 1110, 1111 \}, 
                      P_4=\{1100, 1101, 1010, 1011\}\}.
\end{align*}
For $t=2$, we have a $(\mathcal{P}, t)$-encoding $\mathcal{C}_{\mathcal{P}}: \mathbb{F}_2^4 \rightarrow \mathbb{F}_2^{10}$ defined as
$$\mathcal{C}_{\mathcal{P}}(u)=(u, u_p), \ \text{where} \ u_p=\begin{cases}
000000 & \ \text{if} \ u \in P_1 \\
111100 & \ \text{if} \ u \in P_2 \\
001111 & \ \text{if} \ u \in P_3 \\
110011 & \ \text{if} \ u \in P_4 \\
\end{cases}.$$
This encoding $\mathcal{C}_{\mathcal{P}}$ is an $(f^{(i)},t)$-FCC for all $i\in [3]$.
\end{example}

 Since the join of partitions is a refinement of each individual partition, 
the following theorem follows directly from Lemma~\ref{lem1}.

\begin{theorem}\label{thm_join}
Let $f^{(1)}, f^{(2)}, \ldots, f^{(K)}$ be functions on $\mathbb{F}_q^k$ with 
domain partitions 
$\mathcal{P}_{f^{(1)}}, \mathcal{P}_{f^{(2)}}, \ldots, \mathcal{P}_{f^{(K)}}$, 
respectively. 
If
$$
\mathcal{P}=\mathcal{P}_{f^{(1)}} \vee 
\mathcal{P}_{f^{(2)}} \vee \cdots \vee 
\mathcal{P}_{f^{(K)}},
$$
then any $(\mathcal{P}, t)$-encoding is an $(f^{(i)}, t)$-FCC for every 
$i\in[K]$.
\end{theorem}

Theorem~\ref{thm_join} shows that a single FCC can be used for multiple functions simultaneously, even when their domain partitions are not same.
The next example illustrates this.

\begin{example}\label{ex1}
Consider a scenario with one transmitter and two receivers, where each receiver requires protection for one Hamming weight distribution function: the first receiver for $\Delta_6$ and the second receiver for $\Delta_9$, where  
$$\Delta_T(u)=\left \lfloor \frac{\text{wt}(u)}{T} \right \rfloor \quad \text{for all } \ u \in \mathbb{F}_2^k .$$
For $k=35$, the corresponding domain partitions are
\begin{align*}
\mathcal{P}_{\Delta_6}&=\{\{u \in \mathbb{F}_2^{35} \mid 0\le \text{wt}(u) \le 5\}, \{u \in \mathbb{F}_2^{35} \mid 6\le \text{wt}(u) \le 11\}, \{u \in \mathbb{F}_2^{35} \mid 12\le \text{wt}(u) \le 17\}, \\
&\{u \in \mathbb{F}_2^{35} \mid 18\le \text{wt}(u) \le 23\}, \{u \in \mathbb{F}_2^{35} \mid 24\le \text{wt}(u) \le 29\}, \{u \in \mathbb{F}_2^{35} \mid 30\le \text{wt}(u) \le 35\}\},
\end{align*}
and
\begin{align*}
\mathcal{P}_{\Delta_9}&=\{\{u \in \mathbb{F}_2^{35} \mid 0\le \text{wt}(u) \le 8\}, \{u \in \mathbb{F}_2^{35} \mid 9\le \text{wt}(u) \le 17\}, \\
&\{u \in \mathbb{F}_2^{35} \mid 18\le \text{wt}(u) \le 26\},  \{u \in \mathbb{F}_2^{35} \mid 27\le \text{wt}(u) \le 35\}\}.
\end{align*}
 As proved in \cite{GXZZ2025}, there exists an optimal $(\Delta_T, t)$-FCC with redundancy $2t$ for a Hamming weight distribution function $\Delta_T$ if $T\ge t+1$.  Therefore, for $t=2$, we have the following solutions.

\textbf{Solution 1:} Use an ECC for all $k=35$ symbols that can correct up to $t=2$ error for which the minimum required length is $46$ \cite{CodeTable}.

\textbf{Solution 2:} Append the redundancy of $(\Delta_6, t)$-FCC and $(\Delta_9, t)$-FCC with the message and send it on the shared link. This way each receiver will get protection for its desired function and the total length of the transmission will be $k+2t+2t=35+4+4=43.$

\textbf{Solution 3:} Use Theorem \ref{thm_join} and find the coarsest common refinement of $\mathcal{P}_{\Delta_6}$ and $\mathcal{P}_{\Delta_9}$ given as
\begin{align*}
\mathcal{P}_{\Delta_6} \vee \mathcal{P}_{\Delta_9} &=\{\{u \in \mathbb{F}_2^{35} \mid 0\le \text{wt}(u) \le 2\}, \{u \in \mathbb{F}_2^{35} \mid 3\le \text{wt}(u) \le 5\}, \{u \in \mathbb{F}_2^{35} \mid 6\le \text{wt}(u) \le 8\}, \\
&\{u \in \mathbb{F}_2^{35} \mid 9\le \text{wt}(u) \le 11\}, \{u \in \mathbb{F}_2^{35} \mid 12\le \text{wt}(u) \le 14\}, \{u \in \mathbb{F}_2^{35} \mid 15\le \text{wt}(u) \le 17\}, \\
&\{u \in \mathbb{F}_2^{18} \mid 0\le \text{wt}(u) \le 20\}, \{u \in \mathbb{F}_2^{35} \mid 21\le \text{wt}(u) \le 23\}, \{u \in \mathbb{F}_2^{35} \mid 24\le \text{wt}(u) \le 26\}, \\
&\{u \in \mathbb{F}_2^{35} \mid 27\le \text{wt}(u) \le 29\}, \{u \in \mathbb{F}_2^{35} \mid 30\le \text{wt}(u) \le 32\}, \{u \in \mathbb{F}_2^{35} \mid 33\le \text{wt}(u) \le 35\}\}, \\
& =\mathcal{P}_{\Delta_3}. 
\end{align*}

Since $3\ge t+1$ for $t=2$, we have an optimal $(\Delta_3, t)$-FCC with redundancy $2t=4$ which is a $(\mathcal{P}_{\Delta_3}, t)$-encoding. Therefore, we can use this encoding by which the desired functions of both receivers will be protected from two errors, and  the length will be $k+2t=35+4=39.$


\end{example}

The savings achieved by employing a single code to protect multiple partitions (or functions) can be quantified in two ways: (i) the saved redundancy per function, referred to as the partition redundancy gain; and (ii) the relative improvement in the code rate, referred to as the relative rate increment.

\begin{definition}[Partition gains]
Let $\mathcal{P}_{1}, \mathcal{P}_{2}, \ldots, \mathcal{P}_{K}$ be partitions of $\mathbb{F}_q^k$, and let
$\mathcal{P} = \mathcal{P}_1 \vee \mathcal{P}_2 \vee \cdots \vee \mathcal{P}_K$.
The \emph{partition redundancy gain} of a $(\mathcal{P},t)$-encoding with redundancy $r$ is defined as
$$
\frac{1}{K}\left(\sum_{i=1}^{K} r_{\mathcal{P}_i}(k,t) - r\right).
$$
The \emph{relative rate increment} is defined as
$$
\frac{R_{\mathrm{new}} - R_{\mathrm{old}}}{R_{\mathrm{old}}}
= \frac{\sum_{i=1}^{K} r_{\mathcal{P}_i}(k,t) - r}{k+r},
$$
where $R_{\mathrm{new}}=\frac{k}{k+r}$ and $R_{\mathrm{old}}=\frac{k}{k+\sum_{i=1}^{K} r_{\mathcal{P}_i}(k,t)}.$
\end{definition}

These partition gains can be seen in terms of functions also. 
   Let $f^{(1)}, f^{(2)}, \ldots, f^{(K)}$ be functions on $\mathbb{F}_q^k$ with domain partitions $\mathcal{P}_{f^{(1)}}, \mathcal{P}_{f^{(2)}},$ $\ldots, $ $\mathcal{P}_{f^{(K)}}$, respectively, and $\mathcal{P} = \mathcal{P}_{f^{(1)}} \vee \mathcal{P}_{f^{(2)}} \vee \cdots \vee \mathcal{P}_{f^{(K)}}$. Then the partition redundancy gain of a $(\mathcal{P},t)$-encoding with redundancy $r$ is  
    $$\frac{1}{K}\left(\sum_{i=1}^{K} r_{\mathcal{P}_{f^{(i)}}}(k,t) - r\right) = \frac{1}{K} \left(\sum_{i=1}^{K} r_{f^{(i)}}(k,t) - r\right).$$
    In a similar manner, the relative rate increment is $\frac{\sum_{i=1}^{K} r_{f^{(i)}}(k,t) -r}{k+r}$.

For the code given in solution 3 of Example \ref{ex1}, the partition redundancy gain is 
$\frac{1}{2}((4+4)-4)=2,$ and the relative rate  increment is $\frac{4}{39},$
which are the maximum possible gains as $(\mathcal{P}_{\Delta_3}, 2)$-encoding used here is optimal with $r_{\mathcal{P}_{\Delta_3}}(35, 2)=4$.

The bounds on the redundancy of the join of partitions are given in the following theorem.
\begin{theorem}\label{thm:UB}
    Let $\mathcal{P}_{1}, \mathcal{P}_{2},$ $\ldots, $ $\mathcal{P}_{K}$ be partitions of $\mathbb{F}_q^k$. Then 
    $$ \max_{i\in [K]}(r_{\mathcal{P}_i}(k,t)) \leq r_{\mathcal{P}_1 \vee \mathcal{P}_2 \vee \cdots \vee \mathcal{P}_K}(k,t) \leq \min \left( N_q(q^k,2t+1)-k, \sum_{i=1}^{K} r_{\mathcal{P}_i}(k,t) \right),$$
    where $N_q(M,d)$ denotes the minimum length of an ECC with $M$ codewords and minimum distance $d$.
\end{theorem}

\begin{proof}
Since the join of partitions is a refinement of each individual partition, Lemma~\ref{lem1} implies that
$$
r_{\mathcal{P}_i}(k,t) \;\le\; r_{\mathcal{P}_1 \vee \mathcal{P}_2 \vee \cdots \vee \mathcal{P}_K}(k,t)
\qquad \text{for all } i \in [K].
$$
Taking the maximum over $i \in [K]$ gives the desired lower bound.

The first part of the upper bound is quite straightforward, since any $(n,q^k,2t+1)$
ECC is an $(\mathcal{Q},t)$-encoding for the finest 
partition $\mathcal{Q}=\{\{x\}: x\in\mathbb{F}_q^k\}$, which is a refinement 
of any partition. Hence, by Lemma~\ref{lem1},
$$
r_{\mathcal{P}_1 \vee \mathcal{P}_2 \vee \cdots \vee \mathcal{P}_K}(k,t)
\le N_q(q^k,2t+1)-k.
$$

 For the other part, let $C:\mathbb{F}_q^k \to \mathbb{F}_q^{k+\sum_{i=1}^{K}r_{\mathcal{P}_i}(k,t)}$ be an encoding defined as 
    $$C(u)=(u,r_1(u), r_2(u), \ldots, r_K(u)),$$
    where $r_i(u)$ is the redundancy vector of a $(\mathcal{P}_i, t)$-encoding of length $r_{\mathcal{P}_i}(k,t)$.

    Let $u\in U$ and $v \in V$, where $U$ and $V$ are two different sets in $\mathcal{P}=\mathcal{P}_1 \vee \mathcal{P}_2 \vee \cdots \vee \mathcal{P}_K$. Then
    \begin{equation}\label{dist1}
        d(C(u), C(v)) \geq d((u,r_i(u)), (v,r_i(v))), \quad \forall i \in [K].
    \end{equation}
    By the definition of join of partitions, we have $U=U_1 \cap U_2 \cap \cdots \cap U_K$ and $V=V_1 \cap V_2 \cap \cdots \cap V_K$ for some $U_i, V_i \in \mathcal{P}_i, i \in [K]$. That means $u \in U_i$ and $v \in V_i$ for all $i \in [K]$. Since sets $U$ and $V$ are different, there must be one $j \in [K]$ for which $U_j \neq V_j$. Then by the definition of $(\mathcal{P}_j, t)$-encoding,
    $$d((u, r_j(u)), (v, r_j(v))) \geq 2t+1,$$
    and from \eqref{dist1}, we have $d(C(u), C(v)) \geq 2t+1$. Therefore, $C$ is a $(\mathcal{P}, t)$-encoding, and hence
$$
    r_{\mathcal{P}}(k,t) \leq \sum_{i=1}^{K} r_{\mathcal{P}_i}(k,t).
$$
\end{proof}

\subsection{FCPCs for linear functions}
If  $f:\mathbb{F}_q^k \rightarrow \mathbb{F}_q^{\ell}$ is a linear function, then the domain partition of this function is the coset partition of its kernal, as the kernal of a linear map is a subspace of the domain. Therefore, we have the following result. 

\begin{lemma}
Let $U$ be a subspace of $\mathbb{F}_q^k$ and $\mathcal{U}$ denote the coset partition of $U$. Then for a positive integer $t$, an $(\mathcal{U}, t)$-encoding is an $(f,t)$-FCC for any linear function $f:\mathbb{F}_q^k \rightarrow \mathbb{F}_q^{\ell}$ with Ker$(f) = U$.
\end{lemma}

From Theorem \ref{thm_join}, we can get the following result.

\begin{theorem}\label{thm_join_l}
 If $f^{(1)}, f^{(2)}, \ldots, f^{(K)}$ are linear functions from $\mathbb{F}_q^k$ to $\mathbb{F}_q^{\ell}$.  
Then a $(\mathcal{U}, t)$-encoding  is an $(f^{(i)}, t)$-FCC for all $i\in [K]$, if $\mathcal{U}$ is the coset partition of the subspace $U=\ker(f^{(1)}) \cap\ker(f^{(2)})  \cap \cdots \cap \ker(f^{(K)})$.
\end{theorem}

\begin{proof}
   Let $\mathcal{P}_i$ be the coset partition of $\ker (f^{(i)})$, for all $i\in [K]$. We need to show that
   $$\mathcal{U}=\mathcal{P}_1 \vee \mathcal{P}_2 \vee \cdots \vee \mathcal{P}_K,$$
   where $\mathcal{U}$ is the coset partition of $U=\ker(f^{(1)}) \cap \ker(f^{(2)}) \cap \cdots \cap \ker(f^{(K)}).$

   Let $X \in \mathcal{P}_1 \vee \mathcal{P}_2 \vee \cdots \vee \mathcal{P}_K$ and $X\neq \emptyset$. Then 
   \begin{equation}\label{eq_X}
       X=\left( x_1 + \ker(f^{(1)}) \right) \cap \left( x_2 + \ker(f^{(2)}) \right) \cap \cdots \cap \left( x_K + \ker(f^{(K)}) \right),
   \end{equation}
   for some $x_i \in \mathbb{F}_q^k, i \in [K]$. Since $X\neq \emptyset$, we have $x \in X$. This means
   $$x \in x_i+ \ker(f^{(i)})\quad \forall i \in [K].$$
   Thus, $x_i +\ker(f^{(i)}) = x + \ker(f^{(i)})$ for all $i \in [K]$. Then from \eqref{eq_X}, we have $$X= \cap_{i=1}^{K} \left( x + \ker(f^{(i)}) \right) = x+ \cap_{i=1}^{K} \ker(f^{(i)}) = x+U \in \mathcal{U}.$$
   On the other hand, let $y + U \in \mathcal{U}$. Then 
   $$y + U = y + \cap_{i=1}^{K}\ker(f^{(i)}) = \cap_{i=1}^{K} \left( y + \ker(f^{(i)}) \right).$$
   Since each $y + \ker (f^{(i)})$ is a block in $\mathcal{P}_i$ for $i\in [K]$, we have
   $$y + U \in \mathcal{P}_1 \vee \mathcal{P}_2 \vee \cdots \vee \mathcal{P}_K.$$
   Therefore, $\mathcal{U} = \mathcal{P}_1 \vee \mathcal{P}_2 \vee \cdots \vee \mathcal{P}_K$. From Theorem \ref{thm_join}, we get that $(\mathcal{U}, t)$-encoding is also an $(f^{(i)}, t)$-FCC.
\end{proof}

The following example demonstrates a construction of single code for two linear functions.
\begin{example}\label{ex4}
    Consider the two functions $f_1, f_2: \mathbb{F}_2^3 \to \{ 0, 1\}$ defined as
    $$f_1(x_1, x_2, x_3)=x_1 \quad \text{and} \quad f_2(x_1, x_2, x_3)= x_2.$$
    Clearly, $\ker(f_1)=\{ 000, 010, 001, 011\}$ and $\ker(f_2)=\{ 000, 100, 001, 101\}$. We have $U= \ker(f_1) \cap \ker(f_2)= \{ 000,001 \}$, and its coset partition is $\mathcal{U} = \{U_1=\{000, 001\}, U_2=\{100, 101\}, U_3=\{010, 011\}, U_4=\{110,111\}\}$. For $t=1$, we have the following $(\mathcal{U}, t)$-encoding with $3$ redundancy.
   \[
\begin{array}{c|c|c|c}
\text{Block in }\mathcal{U} & \text{Message vector} & \text{Redundancy} & \text{Codeword} \\
\hline
U_1 & 000 & 000 & 000000 \\
U_1 & 001 & 000 & 001000 \\
U_2 & 100 & 110 & 100110 \\
U_2 & 101 & 110 & 101110 \\
U_3 & 010 & 101 & 010101 \\
U_3 & 011 & 101 & 011101 \\
U_4 & 110 & 011 & 110011 \\
U_4 & 111 & 011 & 111011
\end{array}
\]
This systematic code works as both $(f_1, 1)$-FCC and $(f_2, 1)$-FCC. If we construct an FCC only for $f_1$, then we need only $2$ redundancy symbols: assign $00$ to all vectors in $\ker(f_1)$ and $11$ to all vectors in $100+\ker(f_1)$. This $(f_1, 1)$-FCC is optimal, i.e., $r_{f_1}(3,1)=2$. Similarly, we have $r_{f_2}(3,1)=2$. Therefore, the partition redundancy gain of $(\mathcal{U}, t)$-encoding given above is $\frac{1}{2}(2+2-3)=\frac{1}{2}$, and the relative rate increment is $\frac{1}{6}$.
\end{example}



\subsection{FCPC construction for locally bounded partitions}

The main advantage of FCPCs is that the encoding depends only on the
domain partition and not on a specific function. Consequently, a single encoding can
be used for multiple functions that induce the same partition on the domain. From
this viewpoint, existing constructions of FCCs in the
literature can also be interpreted as constructions for the partitions induced by the
corresponding functions. When formulated in terms of partitions, such constructions
have the additional flexibility that the same encoding works simultaneously for
multiple functions attaining that partition, or for functions whose domain partitions
have that partition as their join.

In this subsection, we illustrate this perspective through a known FCC construction
for locally binary functions presented in \cite{LBWY2023}. We show that the same idea
can be reformulated as an FCPC construction by replacing function values
with a fixed ordering of partition blocks. Thus, the construction presented in this
subsection can be viewed as a reformulation of the FCC construction in
\cite{LBWY2023} in terms of partitions rather than functions.

\begin{definition}[Locally $(\rho, \lambda)$-bounded partition]
Let $\mathcal{P}=\{P_1,P_2,\ldots,P_E\}$
be a partition of $\mathbb{F}_q^k$.  
 $\mathcal{P}$ is said to be a \emph{locally $(\rho,\lambda)$-bounded partition} if for every
$u\in\mathbb{F}_q^k$,
$$
\bigl|\{ j\in[E] : B(u,\rho)\cap P_j \neq \emptyset \}\bigr| \le \lambda,
$$
where $B(u,\rho)$ denotes the Hamming ball of radius $\rho$ centered at $u$.
\end{definition}

For $\lambda=2$, every Hamming ball of radius $\rho$ intersects at most two blocks of
the partition $\mathcal{P}$. 
We now specialize to $(2t,2)$-locally bounded partitions and present a $(\mathcal{P},t)$-encoding for it. 
\begin{construction}[$(\mathcal{P},t)$-encoding for locally $(2t,2)$-bounded partitions]
\label{cons:Pt-encoding}
Let $\mathcal{P}=\{P_1,P_2,\ldots,P_E\}$ be a $(2t,2)$-locally bounded partition of
$\mathbb{F}_q^k$. Fix an ordering of the blocks of $\mathcal{P}$.
Define the encoding map $\mathcal{C} : \mathbb{F}_q^k \rightarrow \mathbb{F}_q^{k+2t}
$ as follows. For $u\in P_i \subseteq \mathbb{F}_q^k$, set
$$
\mathcal{C}(u)=(u,u_p), \quad \text{where} \quad 
u_p=
\begin{cases}
(\underbrace{0,0,\ldots,0}_{2t \ \text{times}}), & \text{if } i=\min B_{\mathcal{P}}(u,2t),\\[2mm]
(\underbrace{1,1,\ldots,1}_{2t \ \text{times}}), & \text{otherwise},
\end{cases}
$$
where $B_{\mathcal{P}}(u,2t)=\{ j\in[E] : B(u,2t)\cap P_j \neq \emptyset \}.$
\end{construction}

\begin{theorem}\label{thm:Pt-encoding-correctness}
Let $\mathcal{P}$ be a $(2t,2)$-locally bounded partition of $\mathbb{F}_q^k$, and
let $\mathcal{C}$ be the encoding map defined above in Construction~\ref{cons:Pt-encoding}.
Then $\mathcal{C}$ is a $(\mathcal{P},t)$-encoding.
\end{theorem}

\begin{proof}
Fix $u\in P_i$ and $v\in P_j$ with $i\neq j$. We have following two cases.

\noindent
\textbf{Case 1: $d(u,v)\ge 2t+1$.}
Since $\mathcal{C}(u)=(u,u_p)$ and $\mathcal{C}(v)=(v,v_p)$, we have
$$
d\bigl(\mathcal{C}(u),\mathcal{C}(v)\bigr)=d(u,v)+d(u_p, v_p)\ge d(u,v)\ge 2t+1.
$$

\noindent
\textbf{Case 2: $d(u,v)\le 2t$.}
Then $v\in B(u,2t)$ and $u\in B(v,2t)$, so
$$
B(u,2t)\cap P_i\neq\emptyset,\qquad B(u,2t)\cap P_j\neq\emptyset.
$$
Thus $\{i,j\}\subseteq B_{\mathcal{P}}(u,2t)$. Since $\mathcal{P}$ is a locally $(2t,2)$-bounded partition, we have $|B_{\mathcal{P}}(u,2t)|\le 2$, and therefore
$
B_{\mathcal{P}}(u,2t)=\{i,j\}.
$
Similarly, we have
$B_{\mathcal{P}}(v,2t)=\{i,j\}$. Without loss of generality, assume $i<j$. Then, by the definition of $\mathcal{C}$, we have $u_p=(0,0,\ldots,0)\in \mathbb{F}_q^{2t}$ and $v_p=(1,1,\ldots,1)\in \mathbb{F}_q^{2t}$.
Therefore, $d(u_p,v_p)=2t$ and
$$
d\bigl(\mathcal{C}(u),\mathcal{C}(v)\bigr)
= d(u,v)+d(u_p,v_p)
\ge  2t+1,
$$
as $d(u,v)\ge 1$ since $u\neq v$.

Therefore, $\mathcal{C}$ is a $(\mathcal{P}, t)$-encoding.
\end{proof}

\subsection{FCPC construction for weight based partitions}

In this subsection, we consider the structured partitions of $\mathbb{F}_q^k$ which are based on the Hamming weights of the vectors. The \emph{weight partition} is denoted by $\mathcal{W}=\{W_0,W_1,\ldots,W_k\}$, where each block $W_i$ (called the weight block) consists of all the vectors of weight $i$. We next consider a family of coarser partitions called \emph{grouped weight partitions} obtained by merging weight blocks, and defined as follows.

\begin{definition}[Grouped weight partition]
Let $\mathcal{W}=\{W_0,\ldots,W_k\}$ be the weight partition of $\mathbb{F}_q^k$.  
A partition $\mathcal{G}=\{G_1,\ldots,G_m\}$ of $\mathbb{F}_q^k$ is called a
\emph{grouped weight partition} if each block $G_j$ is of the form
$$
G_j = \bigcup_{i \in S_j} W_i,
$$
where $S_1,\ldots,S_m \subseteq \{0,\ldots,k\}$ form a partition of $\{0,\ldots,k\}$.
\end{definition}

A grouped weight partition $\mathcal{G}=\{G_1,\ldots,G_m\}$ is called a
\emph{consecutive grouped weight partition} if each index set $S_j$ forms an integer interval, $S_j = \{a_j, a_j+1, \ldots, b_j\},$ for some $0 \le a_j \le b_j \le k$.
Equivalently, each $G_j$ is a union of consecutive weight blocks.
The weight partition is the finest grouped weight partition.


A construction of FCCs for the Hamming weight distribution functions over $\mathbb{F}_2^k$ using Gray codes \cite{S1997, S2004} was proposed in \cite[Subsection~IV.D]{GXZZ2025}. This construction improves the previously known upper bounds on the redundancy for these functions, as compared to the bounds reported in \cite{LBWY2023}. Moreover, \cite{GXZZ2025} establishes that their $(f,t)$-FCC construction is optimal for the Hamming weight distribution functions $\Delta_T$ whenever $T \ge t+1$.

In this subsection, we follow the same idea of employing Gray codes, but generalize the framework to construct FCPCs for grouped weight partitions. This generalization is in two respects. First, the class of grouped weight partitions strictly includes the domain partitions induced by Hamming weight distribution functions, and hence our results include the earlier constructions given in \cite{GXZZ2025} as special cases. Second, while the prior works \cite{LBWY2023,GXZZ2025} consider binary fields, our constructions are valid over general finite fields $\mathbb{F}_q$.

Furthermore, in Corollary \ref{col:HWDFO}, we extend the optimality condition from $T \ge t+1$ in the binary case to $T \ge \left \lceil \frac{2t+1}{q} \right \rceil$ for $\mathbb{F}_q^k$.

\begin{definition}[Gray codes]
A $q$-ary Gray code of dimension $k$ is an ordering
$(g_0,g_1,\ldots,g_{q^k-1})$ of all vectors in $\mathbb{F}_q^k$
such that consecutive vectors differ in exactly one coordinate,
i.e.,
$$
d(g_i,g_{i+1})=1,\qquad \text{for all } \ 0\le i\le q^k-2.
$$
\end{definition}
If the first and last vectors in the ordering also
differ in exactly one coordinate, the Gray code is said to be
\emph{cyclic}.
Gray codes and their cyclic variants have been extensively studied; see \cite{S1997,M2023} for surveys and general background.
General algorithmic constructions of cyclic $q$-ary Gray codes have been proposed in \cite{S2004}, with explicit examples including a cyclic Gray code on $\mathbb{F}_3^3$.

\begin{example}
A cyclic Gray code on $\mathbb{F}_3^3$
constructed using the generalized Gray code algorithm
in \cite{S2004} is given by the following ordering of the
$27$ distinct vectors:
\begin{align*}
&000,001,002,102,100,101,201,202,200,
210,211,212,012,010,011,111,112,110,120,121,122, 222,220,\\
& 221,021,022,020.
\end{align*}
Each pair of consecutive vectors differs in exactly one
coordinate. Moreover, the last vector $020$ differs from the
first vector $000$ in exactly one coordinate, and hence the
ordering forms a cyclic Gray code on $\mathbb{F}_3^3$.
\end{example}

\begin{construction}[Construction of FCPCs for Grouped Weight Partitions]\label{const2}
Let $\mathcal{G}$ denote the grouped weight partition of $\mathbb{F}_q^k$. 
Let $k'$ be the smallest positive integer such that $ q^{k'} \ge \min (k+1, 2t+1).$
Let $\mathcal{E}$ be a systematic $q$-ary ECC of dimension $k'$ and minimum distance $2t+1$, with minimum possible block length $
N = N_q(q^{k'},2t+1).$
Denote by 
$$\mathcal{E} = \{(v_i,p_i)\}_{i=0}^{q^{k'}-1}$$ the set of codewords of $\mathcal{E}$ in the systematic form, i.e., $v_i \in \mathbb{F}_q^{k'}$ denotes the message part and $p_i \in \mathbb{F}_q^{N-k'}$ denotes the parity part.
Arrange the elements of $\mathbb{F}_q^{k'}$ in cyclic Gray code order and index the codewords accordingly, so that
$$
c_i = (v_i,p_i), \qquad i = 0,1,\ldots,q^{k'}-1,
$$
with
$$
d(v_i,v_{i+1})=1 \quad \text{for } 0\le i\le q^{k'}-2,
\qquad \text{and} \qquad
d(v_{q^{k'}-1},v_0)=1.
$$
Now define the encoding map $\mathfrak{E} : \mathbb{F}_q^k \to \mathbb{F}_q^{k+(N-k')}$
as follows.
    For any $u \in \mathbb{F}_q^k$, define
    $$
    \mathfrak{E}(u) = (u,p_{\mathrm{wt}(u)\bmod (2t+1)}).
    $$
\end{construction}

\begin{proof}[Proof of correctness]
Let $u,v \in \mathbb{F}_q^k$ be such that $u$ and $v$ belong to different blocks in $\mathcal{G}$, which means $\mathrm{wt}(u) \neq \mathrm{wt}(v)$. Then
\begin{align*}
d(\mathfrak{E}(u),\mathfrak{E}(v)) &= d\big((u,p_{\mathrm{wt}(u)\bmod (2t+1)}),(v,p_{\mathrm{wt}(v)\bmod (2t+1)})\big).\\
&= d(u,v) + d(p_{\mathrm{wt}(u)\bmod (2t+1)},p_{\mathrm{wt}(v)\bmod (2t+1)}).
\end{align*}
Since $\mathcal{E}$ has minimum distance $2t+1$, assuming $\mathrm{wt}(u)< \mathrm{wt}(v)$, we have
$$d(c_{\mathrm{wt}(u)},c_{\mathrm{wt}(v)}) 
= d(v_{\mathrm{wt}(u)},v_{\mathrm{wt}(v)}) 
+ d(p_{\mathrm{wt}(u)},p_{\mathrm{wt}(v)}) 
\ge 2t+1.$$
Hence,
\begin{align*}
    d(p_{\mathrm{wt}(u)\bmod (2t+1)},p_{\mathrm{wt}(v)\bmod (2t+1)}) &\geq 2t+1-  d(v_{\mathrm{wt}(u)\bmod (2t+1)},v_{\mathrm{wt}(v)\bmod (2t+1)})\\
    & \geq 2t+1-\sum_{s=\mathrm{wt}(u)}^{\mathrm{wt}(v)-1}  d(v_{s\bmod (2t+1)},v_{s+1\bmod (2t+1)})\\
    & =2t+1 - (\mathrm{wt}(v)-\mathrm{wt}(u)).
\end{align*}
Also, we have $d(u, v) \geq \mathrm{wt}(v) - \mathrm{wt}(u).$ 
Therefore, we obtain
$$
d(\mathfrak{E}(u),\mathfrak{E}(v)) 
\ge (\mathrm{wt}(v)-\mathrm{wt}(u)) + \big(2t+1-(\mathrm{wt}(v)-\mathrm{wt}(u))\big)
= 2t+1.
$$

Therefore, encoding $\mathfrak{E}$ guarantees a minimum distance of at least $2t+1$ between codewords corresponding to vectors from different blocks of $\mathcal{G}$. Hence, the proposed construction yields a valid $t$-error-correcting FCPC for any grouped weight partition.
\end{proof}

Further, following the general framework for the Hamming weight distribution function in \cite[Subsection IV.D]{GXZZ2025}, we construct FCPCs for consecutive grouped weight partitions over $\mathbb{F}_q$.

\begin{construction}[Construction of FCPCs for consecutive grouped weight partitions]
   \label{const3}
 Let 
$\mathcal G=\{G_1,\ldots,G_m\}$ be a consecutive grouped weight partition of $\mathbb{F}_q^k$ with $
G_j = \bigcup_{i \in S_j} W_i,
$ and
$S_j=\{a_j,\ldots,b_j\}$, and $|S_j|\ge T$ for all $j\in\{2,\ldots,m-1\}$.
Define a collection of vectors
$p_0,p_1,\ldots,p_{T-1}\in\mathbb{F}_q^{T-1}$
as follows. For $0\le i\le T-1$, 
$$
p_i = (\underbrace{1,1,\ldots,1}_{i\ \text{times}},\underbrace{0,0,\ldots,0}_{(T-1-i)\ \text{times}}).
$$
Next, consider vectors $q_0,q_1,\ldots,q_{m-1}\in\mathbb{F}_q^{s}$ with condition
$$
d(q_i,q_j)\ge 2t+1 - T|i-j| \qquad \text{for all } i\neq j.
$$
For any vector $u\in\mathbb{F}_q^k$, let $i_{\mathcal{G}}(u)$ denote the index of the group of
$\mathcal{G}$ to which $u$ belongs, i.e., $u\in G_{i_{\mathcal{G}}(u)}.$
We now define the encoding map
$\mathfrak{E}:\mathbb{F}_q^k \rightarrow \mathbb{F}_q^{k+(T-1)+s}$
by
$$
\mathfrak{E}(u)=\big(u, p(u), q_{i_{\mathcal{G}}(u)}\big), \qquad u\in\mathbb{F}_q^k,
$$
where $p(u)=p_{\pi(u)}$, and 
for $u\in G_i$ with $S_i=\{a_i,\ldots,b_i\}$, $\pi(u)$ is defined as
$$
\pi(u)=
\begin{cases}
(\mathrm{wt}(u)-a_m)\bmod T, 
& i=m,\\[0.6em]
(\mathrm{wt}(u)-a_i)\bmod T, 
& 1\le i\le m-1,\ 
\mathrm{wt}(u)\le a_i+T\left\lfloor \dfrac{|S_i|}{T}\right\rfloor-1,\\[0.8em]
T-1-(b_i-\mathrm{wt}(u)), 
& 1\le i\le m-1,\ \text{otherwise}.
\end{cases}
$$

\end{construction}
\begin{proof}[Proof of correctness:]
Fix a block $G_i$ with weight interval $S_i=\{a_i,\ldots,b_i\}$ and define
$$
L_i = T\left\lfloor \frac{|S_i|}{T}\right\rfloor,
\qquad
R_i = |S_i|-L_i \in \{0,1,\ldots,T-1\}.
$$
For blocks $G_i$ with $1\le i\le m-1$, parity vectors are assigned in two stages.
For the first $L_i$ weights in $S_i$, taken in increasing order,
parity vectors $p_0,p_1,\ldots,p_{T-1}$ are assigned repeatedly in cyclic order,
that is, the weight $a_i+t$ is assigned $p_{t\bmod T}$ for $0\le t\le L_i-1$.
If $R_i>0$, the remaining $R_i$ weights are assigned 
backward, by weight, using $p_{T-1}, p_{T-2},\ldots, p_{T-R_i}$, in particular, all the vectors with the largest
weight $b_i$ are assigned $p_{T-1}$, and the vectors with  weight $b_i-1$ are assigned $p_{T-2}$, and so on.

For the last block $G_m$, parity vectors are assigned cyclically in increasing order
$p_0,p_1,\ldots,p_{T-1}$ throughout the entire block.

Consider $u\in G_{\alpha}$ and $v\in G_{\beta}$ with $1\leq \alpha < \beta\leq m$. We first prove that 
$$d(u,v)+d(p(u), p(v)) \geq T(\beta-\alpha).$$

For all $r,s\in\{0,\ldots,T-1\}$, we have $d(p_r,p_s)=|r-s|.$
Also, for all $x,y\in\mathbb{F}_q^k$, we have $d(x,y)\ge |\mathrm{wt}(x)-\mathrm{wt}(y)|.$
Therefore,
$$
d(u,v)+d\big(p(u),p(v)\big)
\ge |\mathrm{wt}(v)-\mathrm{wt}(u)|+|\pi(v)-\pi(u)|
\ge \Big|\,(\mathrm{wt}(v)-\mathrm{wt}(u))-(\pi(v)-\pi(u))\,\Big|,
$$
where the last inequality is $|A|+|B|\ge |A-B|$.
Define a function $\phi$ as
$$
\phi(x)=\mathrm{wt}(x)-\pi(x),\qquad x\in\mathbb{F}_q^k.
$$
Then the above inequality becomes
\begin{equation}\label{eq:reducephi}
d(u,v)+d\big(p(u),p(v)\big)\ \ge\ |\phi(v)-\phi(u)|.
\end{equation}

Next, we show that $\phi$ increases by at least $T$ when moving from one block to the
next. Fix $i\in\{1,\ldots,m-1\}$. 
For every vector $x$ with weight $\mathrm{wt}(x)\in\{a_i,\ldots,a_i+L_i-1\}$, we have
$$
\pi(x) = (\mathrm{wt}(x)-a_i)\bmod T
      = \mathrm{wt}(x)-a_i - T \left\lfloor \frac{\mathrm{wt}(x)-a_i}{T} \right\rfloor,
$$
and hence
$$
\phi(x)=\mathrm{wt}(x)-\pi(x)
      = a_i+T \left\lfloor \frac{\mathrm{wt}(x)-a_i}{T} \right\rfloor
      \le a_i + T\left\lfloor\frac{L_i-1}{T}\right\rfloor.
$$
Since $L_i$ is a multiple of $T$, we have $\left\lfloor\frac{L_i-1}{T}\right\rfloor = \frac{L_i}{T}-1$ and $\phi(x) \leq a_i+T\left(\frac{L_i}{T}-1 \right)=a_i+L_i-T$.
For the vectors with remaining weights $\mathrm{wt}(x)\in\{a_i+L_i,\ldots,b_i\}$ (the last incomplete round),
we have $\pi(x)=T-1-(b_i-\mathrm{wt}(x))$, and thus
$$
\phi(x) = \mathrm{wt}(x)-\bigl(T-1-(b_i-\mathrm{wt}(x))\bigr)=b_i-(T-1).
$$
Since $b_i=a_i+|S_i|-1$ and $R_i=|S_i|-L_i\in\{0,\ldots,T-1\}$, we have
\[
b_i-(T-1)=a_i+L_i-T+R_i \ \ge\ a_i+L_i-T.
\]
Therefore, for all $x\in G_i$ we have $\phi(x)\le b_i-(T-1)$, and hence
\begin{equation}\label{eq:maxphiGi}
\max_{x\in G_i}\phi(x) \leq b_i-(T-1).
\end{equation}

Now consider the next block $G_{i+1}$, whose smallest weight is $a_{i+1}=b_i+1$.
For $y\in G_{i+1}$, we have $\phi(y)=\mathrm{wt}(y)-\pi(y)$.
If $\mathrm{wt}(y)\le a_{i+1}+L_{i+1}-1$, then $\pi(y)=(\mathrm{wt}(y)-a_{i+1})\bmod T$ and hence
$$
\phi(y)= a_{i+1}+T\left\lfloor\frac{\mathrm{wt}(y)-a_{i+1}}{T}\right\rfloor \ge a_{i+1}.
$$
Otherwise, $\pi(y)=T-1-(b_{i+1}-\mathrm{wt}(y))$ (which can occur when $i+1<m$), and $\phi(y)=b_{i+1}-(T-1)\ge a_{i+1}$,
where the last inequality holds because $|S_{i+1}|\ge T$.
Therefore, $\phi(y)\ge a_{i+1}$ for all $y\in G_{i+1}$, and 
\begin{equation}\label{eq:minphiGi1}
\min_{y\in G_{i+1}}\phi(y)\;\ge\;a_{i+1}=b_i+1.
\end{equation}
Combining \eqref{eq:maxphiGi} and \eqref{eq:minphiGi1}, we get for all
$x\in G_i$ and $y\in G_{i+1}$,
$$
\phi(y)-\phi(x) \ge (b_i+1)-\bigl(b_i-(T-1)\bigr) = T.
$$
Applying this across the consecutive blocks
$G_\alpha,G_{\alpha+1},\ldots,G_\beta$ yields
\begin{equation}\label{eq:phigap}
\phi(v)-\phi(u)\ \ge\ T(\beta-\alpha).
\end{equation}
By \eqref{eq:reducephi} and \eqref{eq:phigap}, we have
$$
d(u,v)+d\big(p(u),p(v)\big) \ge |\phi(v)-\phi(u)|
 \ge \phi(v)-\phi(u) \ge T(\beta-\alpha).$$
Finally, we can conclude that 
\begin{align*}
    d(\mathfrak{E}(u), \mathfrak{E}(v)) &=d(u,v)+d(p(u),p(v))+d(q_{\alpha}, q_{\beta})\\
    &\geq T(\beta-\alpha)+ 2t+1 - T |\alpha-\beta|=2t+1.
\end{align*}
Hence, the encoding $\mathfrak{E}$ is an FCPC for $\mathcal{G}$.
\end{proof}

\begin{corollary}\label{col:HWDFO}
    For the case where $T\geq \left\lceil\frac{2t+1}{q}\right \rceil$, Construction \ref{const3} gives optimal FCPCs with the following choice of vectors $q_0, q_1,\ldots, q_{m-1}$, where
    $$q_i=(\underbrace{a,a,\ldots,a}_{(2t+1-T) \ \text{times}}), \quad \text{with} \ a = i \bmod q.$$
\end{corollary}
\begin{proof}
For the given vectors $q_0,q_1,\ldots,q_{m-1}$ and $i\ne j$ with $|i-j|\le q-1$,
we have $i\not\equiv j\pmod q$, and hence
$$
d(q_i,q_j)=2t+1-T \ge 2t+1-T|i-j|.
$$
If $|i-j|\ge q$, then
$$
2t+1-T|i-j| \le 2t+1-Tq \le 0,
$$
where the last inequality holds because $T\ge \left\lceil\frac{2t+1}{q}\right\rceil$.
Therefore, the vectors $q_0,q_1,\ldots,q_{m-1}$ satisfy the distance requirement
of Construction~\ref{const3}.

Consequently, we obtain an FCPC with redundancy
$$
r=(T-1)+(2t+1-T)=2t,
$$
which achieves the lower bound on redundancy, and hence the constructed code is optimal.
\end{proof}

The following example demonstrates an FCPC constructed from Construction \ref{const3} that is optimal for the given parameters.

\begin{example}\label{ex:Const3}
Consider a consecutive grouped weight partition $\mathcal{G}$ of
$\mathbb{F}_3^{12}$ defined as
$$
\mathcal{G}
=
\big\{G_1=
W_0,\;
G_2=W_1\cup W_2\cup W_3,\;
G_3=W_4\cup W_5\cup W_6\cup W_7,\;
G_4=W_8\cup W_9\cup W_{10},\;
G_5=W_{11}\cup W_{12}
\big\}.
$$
Thus, $m=5$, and the corresponding index sets are
$$
S_1=\{0\},\quad
S_2=\{1,2,3\},\quad
S_3=\{4,5,6,7\},\quad
S_4=\{8,9,10\},\quad
S_5=\{11,12\}.
$$
Observe that $|S_j|\ge 3$ for all $j\in\{2,3,4\}$, therefore $T=3$. If $t=3$, then $T \ge \left\lceil \frac{2t+1}{q} \right\rceil 
= \left\lceil \frac{7}{3} \right\rceil = 3,$
as required in Corollary \ref{col:HWDFO} for optimality.
The vectors $p_0,p_1,p_2\in\mathbb{F}_3^{2}$ are defined as
$$
p_0=(0,0),\qquad
p_1=(1,0),\qquad
p_2=(1,1).
$$
Since $T=3$ and $t=3$, we have $2t+1-T=4$.  
Following Corollary~\ref{col:HWDFO}, we define
$
q_i=(a,a,a,a),$ with $  a=i \bmod 3,$
for $i=0,1,\ldots,4$ .  
Hence,
$$
q_0=(0,0,0,0), \
q_1=(1,1,1,1), \
q_2=(2,2,2,2),\
q_3=(0,0,0,0), \
q_4=(1,1,1,1).
$$
It can be easily verified that for all $i\neq j$,
$
d(q_i,q_j)\ge 2t+1-T|i-j|,
$
and therefore the distance condition in Construction~\ref{const3} is satisfied.

The encoding $\mathfrak{E}:\mathbb{F}_3^{12}\to\mathbb{F}_3^{12+2+4}$
is given by
$$
\mathfrak{E}(u)=(u,p_{\pi(u)},q_{i_{\mathcal{G}}(u)}).
$$
The values of $\pi(u)$ across all groups and weight values in this example are listed in Table~\ref{tab:Const3}.

\begin{table}[t]
\centering
\caption{Computation of $\pi(u)$ and $p(u)$ for the consecutive grouped weight partition in
Example~\ref{ex:Const3} with $T=3$.}
\label{tab:Const3}
\renewcommand{\arraystretch}{1.25}
\begin{tabular}{|c|c|c|c|c|}
\hline
\textbf{Block} &
\textbf{Weight interval} &
$\boldsymbol{\mathrm{wt}(u)}$ &
$\boldsymbol{\pi(u)}$ &
$\boldsymbol{p(u)}$
\\ \hline
$G_1$ &
$a_1=0,\ b_1=0,\ |S_1|=1$ &
$0$ &
$2 = T-1-(b_1-\mathrm{wt}(u))$ &
$p_2$
\\ \hline
$G_2$ &
$a_2=1,\ b_2=3,\ |S_2|=3$ &
$1$ & $0=\mathrm{wt}(u)-a_2$ & $p_0$ \\
&
&
$2$ & $1=\mathrm{wt}(u)-a_2$ & $p_1$ \\
&
&
$3$ & $2=\mathrm{wt}(u)-a_2$ & $p_2$
\\ \hline
$G_3$ &
$a_3=4,\ b_3=7,\ |S_3|=4$ &
$4$ & $0=\mathrm{wt}(u)-a_3$ & $p_0$ \\
&
&
$5$ & $1=\mathrm{wt}(u)-a_3$ & $p_1$ \\
&
&
$6$ & $2=\mathrm{wt}(u)-a_3$ & $p_2$ \\
&
&
$7$ & $2=T-1-(b_3-\mathrm{wt}(u))$ & $p_2$
\\ \hline
$G_4$ &
$a_4=8,\ b_4=10,\ |S_4|=3$ &
$8$ & $0=\mathrm{wt}(u)-a_4$ & $p_0$ \\
&
&
$9$ & $1=\mathrm{wt}(u)-a_4$ & $p_1$ \\
&
&
$10$ & $2=\mathrm{wt}(u)-a_4$ & $p_2$
\\ \hline
$G_5$ &
$a_5=11,\ b_5=12,\ |S_5|=2$ &
$11$ & $0=\mathrm{wt}(u)-a_5$ & $p_0$ \\
&
&
$12$ & $1=\mathrm{wt}(u)-a_5$ & $p_1$
\\ \hline
\end{tabular}
\end{table}

As an example, let $u\in G_3$ with $\mathrm{wt}(u)=7$ and
$v\in G_4$ with $\mathrm{wt}(v)=9$.  
Then $p(u)=p_2$, $p(v)=p_1$, and
$$
\mathfrak{E}(u)=(u,p_2,q_3), \qquad
\mathfrak{E}(v)=(v,p_1,q_4).
$$
Hence,
$$
d\big(\mathfrak{E}(u),\mathfrak{E}(v)\big)
= d(u,v)+d(p_2,p_1)+d(q_3,q_4)
\ge 2 + 1 + 4 = 7 = 2t+1,
$$
which guarantees correct decoding under $t=3$ symbol errors.
The redundancy of this FCPC is
$r=(T-1)+(2t+1-T)=2t=6.$
Since this matches the lower bound, the code is optimal.
\end{example}

For the other cases where $T\le \left\lceil\frac{2t+1}{q}\right \rceil$, we can use Construction \ref{const2} to find the desired vectors $q_0, q_1,\ldots, q_{m-1}$.
First choose redundancy vectors $q'_0, q'_1, \ldots, q'_{{m-1}}$ using the Gray code method given in Construction \ref{const2} such that
$
d(q'_i, q'_j) \ge 2t+1 - |i-j|$ for all $i \neq j.$
Then we can construct a new family $q_0, q_1, \ldots, q_{m-1}$ by defining $q_i$ as $T$-times repetition of $q'_i$, i.e., $q_i = (q'_i, q'_i, \ldots, q'_i).$ Then these vectors will satisfy the desired condition in Construction \ref{const3}, and have length $T(N_q(q^{k'}, 2t+1)-k')$, where $k'$ is the smallest integer for which $q^{k'}\geq \min(k+1, 2t+1)$. Therefore, the total redundancy of this code is $r=(T-1)+T(N_q(q^{k'}, 2t+1)-k')= T(N_q(q^{k'}, 2t+1)-k'+1)-1.$


\section{Partition graph}\label{par_graph}

In this section, we associate a graph to any given partition (equivalently, to any function-induced domain partition). We then show that identifying a suitable clique in this graph is equivalent to finding a set of vectors whose optimal $\mathcal{D}$-code in the corresponding PDRM gives the optimal redundancy for an FCPC.

\begin{definition}[Partition graph]
For a partition $\mathcal{P}=\{P_1, P_2, \ldots, P_E\}$ of $\mathbb{F}_q^k$, we define $G_{\mathcal{P}}$ to be an $E$-partite graph with vertex set $V =\mathbb{F}_q^k$, such that two vertices $u \in P_i$ and $v \in P_j$ are connected if $i\ne j$ and $d(u,v)=\min_{u_1 \in P_i, u_2 \in P_j} d(u_1,u_2)$.
\end{definition}

We now give examples of partition graphs on $\mathbb{F}_2^3$.
\begin{example}\label{Ex:PG1}
Consider a partition of $\mathbb{F}_2^3$ given as 
$$\mathcal{P}_1=\{\{000, 100, 010, 001\}, \{110, 101, 011, 111\}\}.$$
Then the corresponding $G_{\mathcal{P}_1}$ is shown in Figure \ref{PG1}.
    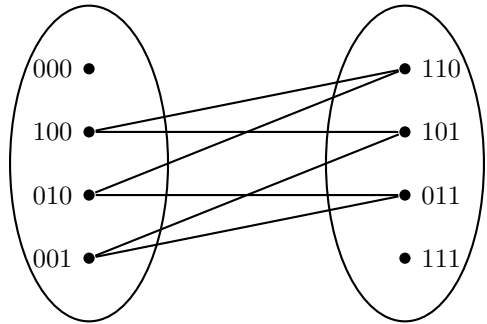
\begin{figure}[!h]
		\centering
    \begin{tikzpicture}[thick,scale=0.7]

\draw[thick] (-3,0) ellipse (1.5 and 3);
\draw[thick] (3,0)  ellipse (1.5 and 3);

\node[circle,fill=black,inner sep=1.5pt,label=left:{$000$}] (a1) at (-3,  1.8) {};
\node[circle,fill=black,inner sep=1.5pt,label=left:{$100$}] (a2) at (-3,  0.6) {};
\node[circle,fill=black,inner sep=1.5pt,label=left:{$010$}] (a3) at (-3, -0.6) {};
\node[circle,fill=black,inner sep=1.5pt,label=left:{$001$}] (a4) at (-3, -1.8) {};

\node[circle,fill=black,inner sep=1.5pt,label=right:{$110$}] (b1) at (3,  1.8) {};
\node[circle,fill=black,inner sep=1.5pt,label=right:{$101$}] (b2) at (3,  0.6) {};
\node[circle,fill=black,inner sep=1.5pt,label=right:{$011$}] (b3) at (3, -0.6) {};
\node[circle,fill=black,inner sep=1.5pt,label=right:{$111$}] (b4) at (3, -1.8) {};

\draw[-] (a2) -- (b1);
\draw[-] (a2) -- (b2);
\draw[-] (a3) -- (b1);
\draw[-] (a3) -- (b3);
\draw[-] (a4) -- (b2);
\draw[-] (a4) -- (b3);
\end{tikzpicture}
\caption{Partition graph for $\mathcal{P}_1$.}
		\label{PG1}
	\end{figure}
Consider another partition of $\mathbb{F}_2^3$ given as 
$$\mathcal{P}_2=\{\{000\}, \{100, 010, 001\}, \{110, 101, 011\}, \{111\}\}.$$
Then the corresponding $G_{\mathcal{P}_2}$ is shown in Figure \ref{PG2}.
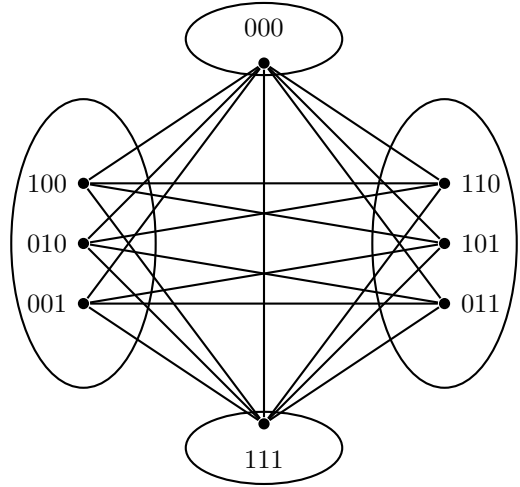
\begin{figure}[!h]
		\centering
\begin{tikzpicture}[thick,scale=0.8]


\node[circle,fill=black,inner sep=1.5pt] (a0) at (0,3) {};
\node[above=4pt of a0] {$000$};
\draw (0,3.4) ellipse (1.3 and 0.6);

\node[circle,fill=black,inner sep=1.5pt] (a4) at (0,-3) {};
\node[below=4pt of a4] {$111$};
\draw (0,-3.4) ellipse (1.3 and 0.6);

\node[circle,fill=black,inner sep=1.5pt,label=left:{$100$}] (l1) at (-3,1.0) {};
\node[circle,fill=black,inner sep=1.5pt,label=left:{$010$}] (l2) at (-3,0.0) {};
\node[circle,fill=black,inner sep=1.5pt,label=left:{$001$}] (l3) at (-3,-1.0) {};
\draw (-3,0) ellipse (1.2 and 2.4);

\node[circle,fill=black,inner sep=1.5pt,label=right:{$110$}] (r1) at (3,1.0) {};
\node[circle,fill=black,inner sep=1.5pt,label=right:{$101$}] (r2) at (3,0.0) {};
\node[circle,fill=black,inner sep=1.5pt,label=right:{$011$}] (r3) at (3,-1.0) {};
\draw (3,0) ellipse (1.2 and 2.4);


\foreach \x in {l1,l2,l3,r1,r2,r3,a4}{
    \draw (a0) -- (\x);
}

\foreach \x in {l1,l2,l3,r1,r2,r3}{
    \draw (a4) -- (\x);
}

\draw (l1) -- (r1);
\draw (l1) -- (r2);
\draw (l2) -- (r1);
\draw (l2) -- (r3);
\draw (l3) -- (r2);
\draw (l3) -- (r3);

\end{tikzpicture}
\caption{Partition graph for $\mathcal{P}_2$.}
		\label{PG2}
	\end{figure} 
\end{example}

For a function $f:\mathbb{F}_q^k \rightarrow S$, the partition graph $G_{\mathcal{P}_f}$ is called its {\it function domain graph}, where $\mathcal{P}_f$ is the domain partition of $f$.

Corollary~\ref{col1} shows that the optimal redundancy of an FCC is completely determined by a suitable set of representatives. For a partition $\mathcal{P}={P_1,\ldots,P_E}$ of $\mathbb{F}_q^k$, by a set of representatives we mean a choice of one vector $u_i \in P_i$ from each block $P_i$, $i=1,\ldots,E$, which is sufficient to compute the optimal redundancy for a $(\mathcal{P},t)$-encoding. The next theorem shows that whenever the partition graph $G_{\mathcal{P}}$ contains a clique of size $E$, the vertices of this clique form such a set of representatives.

\begin{theorem}\label{thm:clique}
Let $\mathcal{P}=\{P_1,\ldots,P_E\}$ be a partition of $\mathbb{F}_q^k$, and let 
$G_{\mathcal{P}}$ be its partition graph.  
If $G_{\mathcal{P}}$ contains a clique $\{u_1,u_2,\ldots,u_E\}$ of size $E$, then this clique
forms a set of representatives for the blocks of $\mathcal{P}$, and
\[
    r_{\mathcal{P}}(k,t)
    = N\!\big(\mathcal{D}_{\mathcal{P}}(t,u_1,u_2,\ldots,u_E)\big).
\]
\end{theorem}

\begin{proof}
Let $\mathcal{C}=\{u_1,\ldots,u_E\}$ be a clique of $G_{\mathcal{P}}$ of size $E$.
To apply Corollary~\ref{col1_p}, we must show:
(i) each block $P_i$ contributes exactly one representative $u_i$, and  
(ii) $\mathcal{D}_{\mathcal{P}}(t,u_1,\ldots,u_E)=\mathcal{D}_{\mathcal{P}}(t;P_1,\ldots,P_E)$.

\textbf{(i)}
Since $G_{\mathcal{P}}$ is an $E$-partite graph with parts $P_1,P_2,\ldots,P_E$, and 
$\mathcal{C}$ is a clique of size $E$, it must contain exactly one vertex from each part.
Thus $\{u_1,\ldots,u_E\}$ forms a complete set of block representatives.

\textbf{(ii)}
Pick $i,j\in[E]$, $i\neq j$.  
Since $u_i$ and $u_j$ lie in a clique, they are adjacent in $G_{\mathcal{P}}$, and by the
definition of the partition graph,
$d(u_i,u_j)
= \min_{u\in P_i, v\in P_j} d(u,v)
= d(P_i,P_j).$
Therefore,
$$
[\mathcal{D}_{\mathcal{P}}(t,u_1,\ldots,u_E)]_{i,j}
= \max(2t+1-d(u_i,u_j),0)
= \max(2t+1-d(P_i,P_j),0)
= [\mathcal{D}_{\mathcal{P}}(t;P_1,\ldots,P_E)]_{i,j}.
$$
All diagonal entries are $0$ in both matrices, so the matrices coincide.  
As both conditions of Corollary~\ref{col1_p} are verified, we obtain
$$
r_{\mathcal{P}}(k,t)
= N\left(\mathcal{D}_{\mathcal{P}}(t;P_1,\ldots,P_E)\right)
= N\left(\mathcal{D}_{\mathcal{P}}(t,u_1,\ldots,u_E)\right),
$$
which completes the proof.
\end{proof}

\begin{example}
Consider the Hamming weight function $\Delta(u)=\text{wt}(u)$ for all $\mathbb{F}_2^3$, then its domain partition is 
$$\mathcal{P}_{\Delta}=\{\{000\}, \{100, 010, 001\}, \{110, 101, 011\}, \{111\}\},$$
which is the same as $\mathcal{P}_2$ given in Example \ref{Ex:PG1}, and the corresponding $G_{\mathcal{P}_{\Delta}}$ is shown in Figure \ref{PG2}. One of its clique is shown in Fig. \ref{PG2.c}, hence we have a set of representative $\{u_1=000, u_2=100, u_3=110, u_4=111\}$ and $r_{\Delta}(k, t) = N_q(\mathcal{D}_f(t, u_1, u_2, u_3, u_4)).$
\begin{figure}[!htbp]
		\centering
    \begin{tikzpicture}[thick,scale=0.5]


\node[circle,fill=black,inner sep=1.5pt] (a0) at (0,3) {};
\node[above=4pt of a0] {$000$};

\node[circle,fill=black,inner sep=1.5pt] (a4) at (0,-3) {};
\node[below=4pt of a4] {$111$};

\node[circle,fill=black,inner sep=1.5pt,label=left:{$100$}] (l2) at (-3,0.0) {};

\node[circle,fill=black,inner sep=1.5pt,label=right:{$110$}] (r2) at (3,0.0) {};


\foreach \x in {l2,r2,a4}{
    \draw (a0) -- (\x);
}

\foreach \x in {l2,r2}{
    \draw (a4) -- (\x);
}

\draw (l2) -- (r2);
\end{tikzpicture}
		\caption{A clique of graph $G_{\mathcal{P}_{\Delta}}$.}
		\label{PG2.c}
	\end{figure}
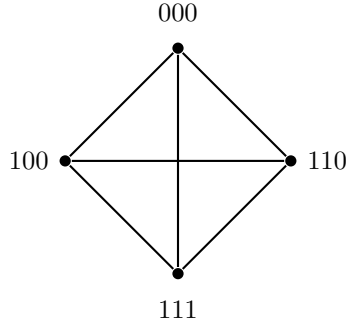 
\end{example}

By using Theorem~\ref{thm:clique}, the problem of finding a $\mathcal{D}$-code for a given partition becomes much simpler. Instead of constructing a solution for all $q^{k}$ vectors in $\mathbb{F}_q^{k}$, it is enough to find a $\mathcal{D}$-code for the vectors that form a full-size clique of size $E$, if such a clique exists. Therefore, establishing the existence of a full-size clique in a partition graph becomes an important problem. This reduction is particularly useful when $E$ is significantly smaller than $q^{k}$.

\subsection{Partition graph for coset partitions}
In this subsection, we study the coset partitions of subspaces of $\mathbb{F}_q^k$, and we provide a sufficient condition under which such partitions admit a full-length clique in the corresponding partition graph. Since the domain partitions of linear functions are exactly the coset partitions of their kernels, the result obtained here naturally applies to linear functions as well. In fact, the sufficient condition given in the lemma below coincides with the condition stated in \cite[Theorem~10]{PR2025} for linear codes. However there also exist non-linear functions whose domain partitions are coset partitions of some subspace of $\mathbb{F}_q^k$, the result holds in a more general setting beyond the linear case. 

We begin with an example demonstrating that a coset partition can arise as the domain partition of a non-linear function, and then proceed to state our sufficient condition.

\begin{example}
Consider the space $\mathbb{F}_2^3$ and its subspace $U = \{000, 110\},$ which has dimension $\dim U = 1 = k-\ell$ with $\ell = 2$.
The cosets of $U$ in $\mathbb{F}_2^3$ are:
\begin{align*}
C_1 &= U = \{000,\,110\}, \qquad \quad
C_2 = e_1 + U = \{100,\,010\},\\
C_3 &= e_3 + U = \{001,\,111\}, \quad
C_4 = e_1+e_3 + U = \{101,\,011\}.
\end{align*}

Now define a function $f:\mathbb{F}_2^3 \to \mathbb{F}_2^2$ by assigning one value to each coset:
$$
f(000)=f(110)=10,\quad 
f(100)=f(010)=01,\quad 
f(001)=f(111)=00,\quad 
f(101)=f(011)=11.
$$
Thus, the domain partition $\mathcal{P}_f$ of $f$ is exactly the coset partition $\{C_1,C_2,C_3,C_4\}$ of $U$.
However, $f$ is not linear. Indeed,
$$
f(110 + 001) = f(111) = 00, \quad \text{but} \quad
f(110) + f(001) = 10 + 00 = 10.
$$
Hence, $f(110+001) \neq f(110) + f(001),$
showing that $f$ is non-linear even though its domain partition is a coset partition of the subspace $U$.
\end{example}

We now present a sufficient condition under which a coset partition admits a full-length clique in its partition graph.

\begin{lemma}\label{lem:coset1}
Let $U \le \mathbb{F}_q^k$ be a subspace with $\dim U = k-\ell$, and coset partition $\mathcal{P} = \{x+U : x \in \mathbb{F}_q^k\}$. If
$$\left|\left\{e_i+U : i \in [k], e_i\notin U \right\}\right| = \ell,$$
where $e_i$ is the $i$-th standard basis vector, then the partition graph $G_{\mathcal{P}}$ has a clique of size $q^\ell$.
\end{lemma}

\begin{proof}
By assumption, there exist $t_1,\dots,t_\ell \in [k]$ such that the cosets
$
U+e_{t_1},U+e_{t_2},\dots,U+e_{t_\ell}
$
are distinct and different from $U$, and every coset of $U$ that contains some standard basis vector $e_i$ is one of these $\ell$ cosets. We prove this lemma using the following two claims.

\noindent\textbf{Claim 1:} If $
a_1 e_{t_1} + a_2 e_{t_2} + \cdots + a_\ell e_{t_\ell} \in U,$
then $a_i = 0$ for all $i \in [\ell]$.

Define the subspace $W = \langle e_{t_1}, e_{t_2}, \ldots, e_{t_\ell} \rangle \subseteq \mathbb{F}_q^k .$
We will show that $W$ is a complementary subspace of $U$, i.e., $$
U + W = \mathbb{F}_q^k
\quad \text{and} \quad
\dim(U) + \dim(W) = k .$$
First, consider any standard basis vector $e_i$ for some $i \in [k]$.
If $e_i \in U$, then clearly $e_i \in U + W$.
If $e_i \notin U$, then by assumption there exists some $j \in [\ell]$ such that $e_i + U = e_{t_j} + U,$
which implies $e_i - e_{t_j} \in U$.
Hence, 
$$e_i = (e_i - e_{t_j}) + e_{t_j} \in U + W .$$
Therefore, $e_i \in U + W$ for all $i \in [k]$, and since
$\{e_1, \ldots, e_k\}$ spans $\mathbb{F}_q^k$, we obtain $U + W = \mathbb{F}_q^k.$
Next, since $\dim(U) = k - \ell$ and $\dim(W) = \ell$, we have $
\dim(U) + \dim(W) = k. $
Thus, $W$ is a complementary subspace of $U$, and consequently $U \cap W = \{0\}.$
Now let
$u = a_1 e_{t_1} + a_2 e_{t_2} + \cdots + a_\ell e_{t_\ell} \in U .$
Since $u \in W$ by construction and $U \cap W = \{0\}$, it follows that
$u = 0$, and hence $a_i = 0$ for all $i \in [\ell]$.
This proves the Claim 1.

Now, for $\alpha = (\alpha_1,\dots,\alpha_\ell)\in \mathbb{F}_q^\ell$, define
$$
v(\alpha) = \sum_{i=1}^\ell \alpha_i e_{t_i} \in \mathbb{F}_q^k,
\qquad
P_{\alpha} = v(\alpha) + U \in \mathcal{P}.
$$
We first show that different $\alpha$ corresponds to different cosets. Suppose $P_{\alpha} = P_{\beta}$ for some $\alpha,\beta\in\mathbb{F}_q^\ell$. Then
$$
v(\alpha) - v(\beta)
= \sum_{i=1}^\ell (\alpha_i - \beta_i) e_{t_i} \in U.
$$
Using Claim 1, we have $\alpha_i - \beta_i = 0$ for all $i \in [\ell]$. Thus $\alpha = \beta$, and the map $\alpha \mapsto P_{\alpha}$ is injective. Since there are $q^\ell$ cosets of $U$ (because $\dim U = k-\ell$), this map is in fact a bijection from $\mathbb{F}_q^\ell$ onto $\mathcal{P}$.

\noindent\textbf{Claim 2:} The set $\mathcal{C}=\{v(\alpha):\alpha\in\mathbb{F}_q^\ell\}$ is a clique in $G_{\mathcal{P}}$.

For each $r\in[\ell]$, define the index set
$$
S_r=\{\,i\in[k]: e_i \in e_{t_r}+U\,\}.
$$
Then $S_1,\dots,S_\ell$ are pairwise disjoint, and for every $i$ with $e_i\notin U$ there exists a unique $r$ such that $i\in S_r$.
For $x=(x_1,\dots,x_k)\in\mathbb{F}_q^k$, define the linear maps
$$
\phi_r(x)=\sum_{i\in S_r} x_i,\qquad r\in[\ell].
$$
We first show that $\phi_r(u)=0$ for all $u\in U$ and $r\in[\ell]$.
Let $u=\sum_{i=1}^k u_i e_i\in U$, which can be written as
$$u = \sum_{r=1}^\ell \sum_{i\in S_r} u_i e_i \;+\; \sum_{i:\,e_i\in U} u_i e_i.$$
For each $r$, we have $e_i-e_{t_r}\in U$ for all $i\in S_r$, hence
$$
u-\sum_{r=1}^\ell \Big(\sum_{i\in S_r}u_i\Big)e_{t_r}
= \sum_{r=1}^\ell \sum_{i\in S_r} u_i (e_i-e_{t_r}) \;+\; \sum_{i:\,e_i\in U} u_i e_i \in U.
$$
Therefore, $\sum_{r=1}^\ell \big(\sum_{i\in S_r}u_i\big)e_{t_r}\in U$, and by Claim~1 we get
$\sum_{i\in S_r}u_i=0$, i.e., $\phi_r(u)=0$ for all $r \in [\ell]$.

Now fix $\alpha\in\mathbb{F}_q^\ell$ and take any $x\in P_\alpha=v(\alpha)+U$, so $x=v(\alpha)+u$ for some $u\in U$. Then
$$
\phi_r(x)=\phi_r(v(\alpha))+\phi_r(u)=\phi_r(v(\alpha)).
$$
Since $t_r\in S_r$ and $v(\alpha)=\sum_{j=1}^\ell \alpha_j e_{t_j}$, we have
$\phi_r(v(\alpha))=\alpha_r$, and hence
\begin{equation}\label{eq:phi_const}
\phi_r(x)=\alpha_r \quad \text{for all }x\in P_\alpha,  r\in[\ell].
\end{equation}
Now for fix $\alpha,\beta\in\mathbb{F}_q^\ell$, set
$
s=\big|\{r\in[\ell]:\alpha_r\neq\beta_r\}\big|.
$
For arbitrary $x\in P_\alpha$ and $y\in P_\beta$, let $z=x-y$. By \eqref{eq:phi_const},
$$
\phi_r(z)=\phi_r(x)-\phi_r(y)=\alpha_r-\beta_r.
$$
Thus, if $\alpha_r\neq\beta_r$ then $\phi_r(z)\neq 0$, which implies that $z$ has at least one nonzero coordinate in the index set $S_r$. Since the sets $S_r$ are disjoint, $z$ has at least one nonzero coordinate in each of the $s$ sets corresponding to $\alpha_r\neq\beta_r$, and therefore $\mathrm{wt}(z)\ge s.$
As $x\in P_\alpha$ and $y\in P_\beta$ were arbitrary, this shows $d(P_\alpha,P_\beta)\ge s$.
On the other hand,
$$
d(v(\alpha), v(\beta)) = \mathrm{wt}\big(v(\alpha)-v(\beta)\big)
=\mathrm{wt}\Big(\sum_{r=1}^\ell (\alpha_r-\beta_r)e_{t_r}\Big)=s,
$$
so we have $d(P_\alpha,P_\beta)=d\big(v(\alpha),v(\beta)\big).$

By the definition of the partition graph, there is an edge between $ v(\alpha)$ and $v (\beta)$. Since $\alpha$ and $\beta $ were taken arbitrarily, the set
$
\mathcal{C} = \{v(\alpha) : \alpha \in \mathbb{F}_q^\ell\}
$
contains exactly one vertex from each coset, and every pair of distinct vertices in $\mathcal{C}$ has an edge. Thus $\mathcal{C}$ is a clique of size $|\mathcal{C}| = q^\ell$, which is a full-size clique in the partition graph of $\mathcal{U}$.
\end{proof}

Therefore, from Lemma \ref{lem:coset1}, if $U \le \mathbb{F}_q^k$ is a subspace with $\dim U = k-\ell$ with coset partition $\mathcal{P} = \{x+U : x \in \mathbb{F}_q^k\}$, and 
$$\left|\left\{e_i+U : i \in [k], e_i\notin U \right\}\right| = \ell,$$
then there exists a set of representative vectors $\{u_1, \ldots, u_{q^{\ell}}\}$ defined in the proof of Lemma \ref{lem:coset1}, and $r_{\mathcal{P}}(k,t)=N_q(\mathcal{D}_f(t, u_1, \ldots, u_{q^{\ell}}))$.

We now show the existence of a full-size clique for partitions defined via coordinate restrictions, and then show that this setting is a special case of Lemma~\ref{lem:coset1}.

\begin{lemma}\label{lem:coset_gen}
Let $J \subseteq [k]$ be a subset of coordinates with $|J| = \ell$, and consider the partition
$
\mathcal{P}_J = \{P_a : a \in \mathbb{F}_q^\ell\}
$
of $\mathbb{F}_q^k$ defined by
$$
P_a = \{x \in \mathbb{F}_q^k : x_J = a\},
$$
where $x_J$ denotes the restriction of $x$ to the coordinates in $J$. Then the partition graph associated with $\mathcal{P}_J$ contains a clique of size $q^\ell$.
\end{lemma}

\begin{proof}
For each $a = (a_j)_{j \in J} \in \mathbb{F}_q^\ell$, define a vector $u^{(a)} \in \mathbb{F}_q^k$ by
$$
u^{(a)}_j =
\begin{cases}
a_j, & j \in J,\\
0,   & j \notin J.
\end{cases}
$$
Clearly $u^{(a)} \in P_a$, since $u^{(a)}_J = a$. Define
$
\mathcal{C} = \{u^{(a)} : a \in \mathbb{F}_q^\ell\}.
$
Then $C$ contains exactly one vertex from each block $P_a$ of the partition $\mathcal{P}_J$. Now take distinct $a,b \in \mathbb{F}_q^\ell$. By construction, $u^{(a)}$ and $u^{(b)}$ may differ only in the coordinates in $J$, and
$$
d(u^{(a)},u^{(b)}) = \left|\{j \in J : a_j \neq b_j\}\right| =s.
$$

We next compute the block distance between $P_a$ and $P_b$. Any $x \in P_a$ and $y \in P_b$ satisfy $x_J = a$ and $y_J = b$, so in each coordinate $j \in J$ with $a_j \neq b_j$ the difference $x-y$ has a nonzero entry. Thus
$$
d(x,y) \ge s, \qquad \text{for all } x \in P_a,\ y \in P_b.
$$
On the other hand, for the specific pair $(x,y) = (u^{(a)},u^{(b)})$ we have $d(u^{(a)},u^{(b)}) = s$. Therefore
$$
d(P_a,P_b)
= \min\{d(x,y) : x \in P_a,\ y \in P_b\}
= s
= d(u^{(a)},u^{(b)}).
$$

Therefore, $u^{(a)}$ and $u^{(b)}$ are adjacent for every distinct $a,b \in \mathbb{F}_q^\ell$, and $\mathcal{C}$ is a clique of size $q^\ell$.
\end{proof}

\begin{remark}
Lemma~\ref{lem:coset_gen} is a special case of Lemma~\ref{lem:coset1}, as for any $J \subseteq [k]$ with $|J|=\ell$, the set
$$
U = P_{\mathbf{0}} = \{x \in \mathbb{F}_q^k : x_J = \mathbf{0}\}
$$
is a subspace of $\mathbb{F}_q^k$ of dimension $k-\ell$. Moreover, $e_i \notin U$ if and only if $i \in J$, and hence the cosets $U+e_i$ with $e_i \notin U$ are precisely $\{U+e_j : j \in J\}$, which are $\ell$ in number. Therefore, the coset partition of $U$ satisfies the condition of Lemma~\ref{lem:coset1}, and the partition $\mathcal{P}_J$ coincides with the coset partition of $U$.
\end{remark}

We now give an example illustrating Lemma \ref{lem:coset_gen}.

\begin{example}
Let $q=3$, $k=3$, and choose $J=\{2,3\}$ so that $|J|=\ell=2$.
The partition $\mathcal{P}_J=\{P_a : a\in\mathbb{F}_3^2\}$, where $P_a=P_{(a_2,a_3)} = \{(x_1, a_2, a_3) : x_1 \in \mathbb{F}_3\}.$
Explicitly, the nine blocks are:
\begin{align*}
P_{(0,0)} &= \{(0,0,0),(1,0,0),(2,0,0)\},
P_{(0,1)} = \{(0,0,1),(1,0,1),(2,0,1)\},
P_{(0,2)} = \{(0,0,2),(1,0,2),(2,0,2)\},\\
P_{(1,0)} &= \{(0,1,0),(1,1,0),(2,1,0)\},
P_{(1,1)} = \{(0,1,1),(1,1,1),(2,1,1)\},
P_{(1,2)} = \{(0,1,2),(1,1,2),(2,1,2)\},\\
P_{(2,0)} &= \{(0,2,0),(1,2,0),(2,2,0)\},
P_{(2,1)} = \{(0,2,1),(1,2,1),(2,2,1)\},
P_{(2,2)} = \{(0,2,2),(1,2,2),(2,2,2)\}.
\end{align*}

Following Lemma~\ref{lem:coset_gen}, the vectors $\mathcal{C}
=\{(0,a_2,a_3): a_2,a_3\in\mathbb{F}_3\}
$
form a clique of size $3^2=9$ in the partition graph $G_{\mathcal{P}_J}$, refer to Fig. \ref{fig:ex8}.

\begin{figure}[ht]
\centering
\begin{tikzpicture}[
    >=latex,
    every node/.style={circle,draw,inner sep=2pt,font=\scriptsize},
    scale=1.2
]

\node (v1) at (90:3cm) {(0,0,0)};
\node (v2) at (50:3cm) {(0,0,1)};
\node (v3) at (10:3cm) {(0,0,2)};
\node (v4) at (-30:3cm) {(0,1,2)};
\node (v5) at (-70:3cm) {(0,1,1)};
\node (v6) at (-110:3cm) {(0,1,0)};
\node (v7) at (-150:3cm) {(0,2,0)};
\node (v8) at (-190:3cm) {(0,2,1)};
\node (v9) at (-230:3cm) {(0,2,2)};

\foreach \u in {1,...,9} {
    \foreach \v in {1,...,9} {
        \ifnum\u<\v
            \draw (v\u) -- (v\v);
        \fi
    }
}

\end{tikzpicture}

\caption{The nine vectors $\mathcal{C}=\{(0,a_2,a_3): a_2,a_3\in\mathbb{F}_3\}$,
one from each block $P_{(a_2,a_3)}$, forming a clique of size $9$ in the partition
graph $G_{\mathcal{P}_J}$.}
\label{fig:ex8}
\end{figure}
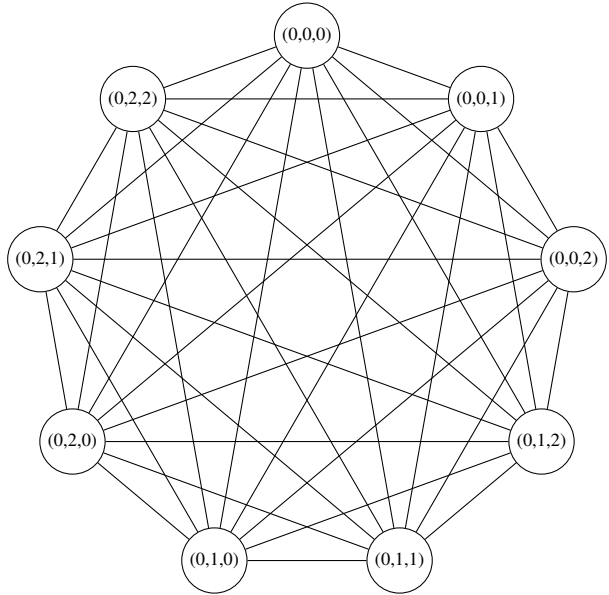

\end{example}

\subsection{Partition Graphs of Weight-Based Partitions}



In this subsection, we study the partition graphs associated with partitions of $\mathbb{F}_q^k$ based on the Hamming weights of the vectors. We first consider the \emph{weight partition}, $\mathcal{W}=\{W_0,W_1,\ldots,W_k\}$. We show that its partition graph contains a full-size clique, and using this property, we derive a lower bound on the redundancy of the corresponding $(\mathcal{W},t)$-FCC over $\mathbb{F}_q$.
We next consider grouped weight partitions and show that for a subclass of it, a full size clique does not exist.

The following lemma shows the existence of a full size clique in the partition graph of the weight partition $\mathcal{W}$.

\begin{lemma}\label{lem:HWF}
The partition graph of the weight partition $\mathcal{W}=\{W_0, W_1, \ldots, W_k\}$ of $\mathbb{F}_q^k$ contains a clique of size $k+1$, where $W_i=\{ u \in \mathbb{F}_q^k  : \mathrm{wt}(u)=i \}$ for all $i\in \{0,1,\ldots, k\}$.
\end{lemma}

\begin{proof}
 Take any $x \in W_i$ and $y \in W_j$, then
$
\mathrm{wt}(x) = i,$ and $  \mathrm{wt}(y) = j.
$
 and 
 $$d(x,y)\geq|\mathrm{wt}(x)-\mathrm{wt}(y)|=|i-j|,$$ 
 and hence $d(W_i, W_j)\geq |i-j|.$
For each $i \in \{0,1,\ldots,k\}$, fix a nonzero symbol $a \in \mathbb{F}_q$ and define
$$
u_i = (\underbrace{a,\ldots,a}_{i\ \text{times}},0,\ldots,0) \in \mathbb{F}_q^k.
$$
Then $\mathrm{wt}(u_i) = i$, so $u_i \in P_i$. For $i \ge j$, the supports of $u_i$ and $u_j$ satisfy
$
\mathrm{supp}(u_j) \subseteq \mathrm{supp}(u_i),
$
and they differ exactly in the last $i-j$ positions, where $u_i$ has $a$ and $u_j$ has $0$. Thus
$
d(u_i,u_j) = |i-j|.
$
Combining with the previous inequality, we obtain
$$
d(W_i,W_j) = d(u_i,u_j)= |i-j|  \qquad \text{for all } i \neq j.
$$

By the definition of the partition graph of $\mathcal{W}$, there is an edge between $u_i \in W_i$ and $u_j \in W_j$ for $i \neq j$, and the set
$
C = \{u_0,u_1,\ldots,u_k\}
$
contains exactly one vertex from each block $W_i$ and forms a clique of size $k+1$.
\end{proof}

Since the Hamming weight function and any function of the form 
$f(x)=g(\mathrm{wt}(x))$, where $g$ is a bijection on $\{0,1,\ldots,k\}$, all have 
the weight partition $\mathcal{W}$ as their domain partition, 
Lemma~\ref{lem:HWF} applies to each of these functions.  
For the special case of the Hamming weight function, a similar result was 
previously presented in \cite[Lemma~6]{LBWY2023}.

\begin{example}\label{ex9}
Consider the weight partition $
\mathcal{W}=\{W_0,W_1,W_2,W_3\}$ of $\mathbb{F}_3^3$.
Fix a nonzero symbol $a\in\mathbb{F}_3$, say $a=1$.  
Define the following vectors, as in the proof of Lemma~\ref{lem:HWF}:
$$
u_0=(0,0,0), \qquad
u_1=(1,0,0), \qquad
u_2=(1,1,0), \qquad
u_3=(1,1,1).
$$
Then $\mathrm{wt}(u_0)=0,\ \mathrm{wt}(u_1)=1,\
\mathrm{wt}(u_2)=2,\ \mathrm{wt}(u_3)=3,$ so $u_i \in W_i$ for all $i=0,1,2,3$. The set $C=\{u_1, u_2, u_3,u_4\}$ forms a full size clique. For $t=2$, the PDRM of the clique vectors is
$$N_3(\mathcal{D}_{\mathcal{W}}(t,u_1,u_2,u_3,u_4))=
\begin{bmatrix}
 0   & 4   & 3   & 2  \\
 4   & 0   & 4   & 3  \\
 3   & 4   & 0   & 4  \\
 2   & 3   & 4   & 0
\end{bmatrix}.
$$
Therefore, instead of constructing $\mathcal{D}$-code for all $27$ vectors in the space, we just need to construct a $\mathcal{D}$-code for this PDRM of $4$ vectors. One such valid $\mathcal{D}$-code is: $z_1= 0000, z_2=1111, z_3=0222, z_4=2001$.
\end{example}

From Lemma \ref{lem:HWF} and Theorem \ref{thm:clique}, we have the following result.
\begin{corollary}\label{col:HWF1}
    For any positive integers $k$ and $t$, and for a fixed $a\in \mathbb{F}_q$, $u_i =(\underbrace{a,\ldots,a}_{i\ \text{times}},0,\ldots,0),\ i=0,1,\ldots,k$,  we have
$$r_{\mathcal{W}}(k,t)=N_q(D_{\mathcal{W}}(t \,; \, u_0,u_1,\ldots,u_k)).$$
\end{corollary}

Using the Plotkin bound in Lemma~\ref{lem:RRHH}, we derive a lower bound on the redundancy of a $(\mathcal{W},t)$-encoding over $\mathbb{F}_q$. An upper bound is 
also obtained by constructing a $(\mathcal{W},t)$-encoding from a suitable ECC.

\begin{corollary}\label{col:HWF2}
    For any positive integers $k$ and $t$, the redundancy of a $(\mathcal{W},t)$-encoding over $\mathbb{F}_q$ satisfies
    $$ \frac{2q}{(k+1)^2(q-1)-a(q-a)}\, S_{k,t} \leq r_{\mathcal{W}}(k,t)\leq N_q(\min(2t+1,k+1),2t),$$
    where $a=(k+1) \bmod q$, $N_q(M,d)$ is the minimum length of an ECC with $M$ codewords and minimum distance $d$, and
    $$S_{k,t}= \begin{cases}
       \displaystyle \frac{k(k+1)(6t+1-k)}{6} &\text{if}\ k\leq 2t \\[4pt]
       \displaystyle \frac{t(2t+1)(3k-2t+1)}{3} &\text {if}\ k\geq 2t
    \end{cases}.$$
\end{corollary}
\begin{proof}
For the vectors $u_i=(\underbrace{a,\ldots,a}_{i},0,\ldots,0)$, the pairwise distance is 
$d(u_i,u_j)=|i-j|$. Hence, PDRM $\mathcal{D}_{\mathcal{W}}(t;u_0,u_1,\ldots,u_k)$ have $(i,j)$-th entry as
$
\mathcal{D}_{i,j}=\max(2t+1-|i-j|,0),\ i\neq j.
$
Summing over the upper triangular part gives
$$
\sum_{i<j} \mathcal{D}_{i,j}
=\sum_{s=1}^{\min(k,2t)} (k+1-s)(2t+1-s)
= S_{k,t},
$$
because there are $k+1-s$ pairs with gap $s=j-i$, and all terms with $s >2t$ vanish. Further solving the expression of $S_{k,t}$, we get
$$
S_{k,t}=
\begin{cases}
\displaystyle \frac{k(k+1)(6t+1-k)}{6}, & k\le 2t,\\[4pt]
\displaystyle \frac{t(2t+1)(3k-2t+1)}{3}, & k\ge 2t.
\end{cases}
$$
Finally, applying the Plotkin bound from Lemma~\ref{lem:RRHH},
$$
N_q(\mathcal{D}_{\mathcal{W}}(t;u_0,u_1,\ldots,u_k))\ge
\frac{2q}{(k+1)^2(q-1)-a(q-a)}
\sum_{i<j} \mathcal{D}_{i,j},
\qquad a=(k+1)\bmod q,
$$
and substituting $\sum_{i<j}\mathcal{D}_{i,j}=S_{k,t}$ gives the lower bound.

\noindent\textbf{Upper bound.}
Set $M = \min(k+1, 2t+1)$, and let $\{p_0,p_1,\ldots,p_{M-1}\}$ be a $q$-ary ECC of length $N_q(M,2t)$, size $M$, and minimum distance $2t$. We now construct a $\mathcal{D}_{\mathcal{W}}(t;u_0,\ldots,u_k)$-code of the same length.

\emph{Case 1: $k+1 \le 2t+1$ (so $M=k+1$).}
Define the codewords $c_i = p_i$ for $i=0,1,\ldots,k$. Then for any $i\neq j$,
$$
d(c_i,c_j) = d(p_i,p_j) \ge 2t.
$$
On the other hand, by the definition of the PDRM,
$$
\mathcal{D}_{i,j} = \max(2t+1-|i-j|,0) \le 2t,
$$
so $d(c_i,c_j)\ge \mathcal{D}_{i,j}$ for all $i\neq j$. Hence $\{c_0,\ldots,c_k\}$ is a $\mathcal{D}_{\mathcal{W}}(t;u_0,\ldots,u_k)$-code of length $N_q(M,2t)$.

\emph{Case 2: $k+1 \ge 2t+1$ (so $M=2t+1$).}
Define
$$
c_i = p_{i \bmod M}, \qquad i=0,1,\ldots,k.
$$
Consider any distinct $i,j$ and let $s=|i-j|$.
\begin{itemize}
    \item If $1\le s \le 2t$, then $s<M$, so $i \not\equiv j \pmod M$ and hence $c_i\neq c_j$. Thus
    $$
    d(c_i,c_j) = d(p_{i\bmod M},p_{j\bmod M}) \ge 2t.
    $$
    Since in this range $\mathcal{D}_{i,j}=2t+1-s\le 2t$, we again have $d(c_i,c_j)\ge \mathcal{D}_{i,j}$.
    \item If $s\ge M=2t+1$, then $\mathcal{D}_{i,j} = \max(2t+1-s,0)=0$. In particular, we are allowed to have $d(c_i,c_j)=0$. Our construction assigns the same codeword exactly when $s$ is a multiple of $M=2t+1$. In that case, $c_i=c_j$ and $d(c_i,c_j)=0=\mathcal{D}_{i,j}$. If $s$ is not a multiple of $M$, then $c_i\neq c_j$ and $d(c_i,c_j)\ge 2t\ge 0 = \mathcal{D}_{i,j}$. Thus in all subcases $d(c_i,c_j)\ge \mathcal{D}_{i,j}$.
\end{itemize}
Therefore, $\{c_0,\ldots,c_k\}$ is a $\mathcal{D}_{\mathcal{W}}(t;u_0,\ldots,u_k)$-code of length $N_q(M,2t)$ also in this case.

Combining both cases, we get
$
N_q(\mathcal{D}_{\mathcal{W}}(t;u_0,\ldots,u_k))
\le N_q(\min(k+1,2t+1),2t),
$
and by Corollary~\ref{col:HWF1} this implies
$$
r_{\mathcal{W}}(k,t) = N_q(\mathcal{D}_{\mathcal{W}}(t;u_0,\ldots,u_k))
\le N_q(\min(k+1,2t+1),2t).
$$
\end{proof}

For $q=3$, $k=3$, and $t=2$, Corollary~\ref{col:HWF2} gives the lower bound $r_{\mathcal{W}}(k,t)\ge 4$. As shown in Example~\ref{ex9}, we obtain a $\mathcal{D}$-code of length $4$ for these parameters. Therefore, the code in that example is optimal, and the bound is tight.

We now see that the full-size clique property does not extend to grouped
weight partitions in general.  The following lemma shows that for a consecutive
grouped weight partition in which at least one middle block contains two or more
consecutive weight blocks, the induced partition graph does not contain a
full-size clique.

\begin{lemma}\label{lem:HWDF}
Let $\mathcal{G}=\{G_1, G_2, \ldots, G_m\}$ be a consecutive grouped weight partition where $G_j=\bigcup_{i \in S_j} W_i$ for each $j$, and assume that 
$|S_{\ell}| \ge 2$ for some $\ell \in \{2,3,\ldots,m-1\}$.   
Then the partition graph of $\mathcal{G}$ contains no clique of size $m$. 
\end{lemma}

\begin{proof}
In the consecutive grouped weight partition 
$\mathcal{G} = \{G_1,G_2,\ldots,G_m\}$, assume that the blocks are ordered so that
their index sets $S_1,\ldots,S_m$ satisfy
$\max S_r < \min S_s$ whenever $r < s.$
For each $j$,
$$
G_j = \bigcup_{i \in S_j} W_i=\{x \in \mathbb{F}_q^k : \min{S_j} \leq \mathrm{wt}(x) \leq \max{S_j}\}.
$$
Let $r,s\in\{1,2,\ldots,m\}$ with $r<s$.  
For any $x\in G_r$ and $y\in G_s$, we have
$$
\min S_r \le \mathrm{wt}(x) \le \max S_r,
\qquad
\min S_s \le \mathrm{wt}(y) \le \max S_s.
$$
Similarly to the proof of Lemma \ref{lem:HWF}, we know that
$d(x,y) \ge  \mathrm{wt}(y)-\mathrm{wt}(x) \geq \min S_s- \max S_r.$
For some non-zero $a\in \mathbb{F}_q$, we have vectors $$x=(\underbrace{a,a, \ldots,a}_{\max S_r\ \text{times}},0,\ldots,0)\in W_{\max S_r} \subseteq G_r \quad \text{and} \quad y= (\underbrace{a,a,\ldots,a}_{\min S_s\ \text{times}},0,\ldots,0)\in W_{\min S_s} \subseteq G_s,$$
for which $d(x,y)=\min S_s - \max S_r.$ Therefore, $d(G_r,G_s) = \min S_s - \max S_r$ for all $r<s$ and $r,s \in \{1,2,\ldots, m\}$.

Now suppose, for a contradiction, that there exists a full-size clique, i.e., a
clique of size $m$ in the partition graph. Then in particular there must be a
triangle with vertices $\{u_r,u_{\ell},u_s\}$, where $u_r \in G_r$, $u_{\ell} \in G_{\ell}$, $u_s \in G_s$ for some $r < \ell < s$. Then we must have
\begin{align*}
d(u_r,u_\ell) &= d(G_r,G_\ell) = \min S_\ell -\max S_r,\\
d(u_\ell,u_s) &= d(G_\ell,G_s) = \min S_s-\max S_\ell,\\
d(u_r,u_s) &= d(G_r,G_s) = \min S_s - \max S_r.
\end{align*}
By the triangle inequality,
$$
d(u_r,u_s)
   \le d(u_r,u_\ell) + d(u_\ell,u_s)
   = \min S_\ell -\max S_r +  \min S_s-\max S_\ell
   = \min S_s - \max S_r - (\max S_\ell -\min S_\ell).
$$
Combining with the required value of $d(u_r,u_t)$, we obtain
$$
\min S_s - \max S_r \;\le\;  \min S_s - \max S_r - (\max S_\ell -\min S_\ell),
$$
that means $\max S_\ell -\min S_\ell \le 0,$ which is impossible as $|S_\ell| \ge 2$. This contradiction shows that no such triple $u_r,u_\ell,u_s$ can exist, so the partition graph contains no triangle involving three distinct blocks $G_r,G_\ell,G_s$. In particular, it
cannot contain a clique of size $m$, since any clique of size $m$ would contain
such a triple.
\end{proof}

\begin{remark}
The domain partition of the Hamming weight distribution function 
$\Delta_T(x)=\left\lfloor \frac{\mathrm{wt}(x)}{T} \right\rfloor$ on $\mathbb{F}_q^k$ is a 
consecutive grouped weight partition 
$\mathcal{G}=\{G_1,\ldots,G_{\lceil k/T\rceil}\}$, where each block is of the 
form $G_j=\bigcup_{i\in S_j} W_i$ with $|S_j|=T$ for all $j$ (except possibly 
the last block when $T\nmid k$).  
Hence, by Lemma~\ref{lem:HWDF}, the domain graph of $\Delta_T$ does not contain 
a full-size clique for any $T \ge 2$.
\end{remark}

   \begin{remark}
There are $2^{k}$ possible consecutive grouped weight partitions of 
$\mathbb{F}_q^k$.  
This is because, the weight partition $\mathcal{W}=\{W_0,\ldots,W_k\}$ has $k$ potential 
boundaries between successive weight blocks $W_{i-1}$ and $W_i$, and each boundary may either be kept (separating the blocks) or removed (merging them).  
Thus, every choice of which boundaries to keep corresponds uniquely to a 
consecutive grouped weight partition, giving $2^k$ possibilities.

In contrast, there are only $k$ distinct Hamming weight distribution functions 
$\Delta_T$ on $\mathbb{F}_q^k$, one for each threshold $T\in\{1,\ldots,k\}$.  
Hence, Lemma~\ref{lem:HWDF} applies to a much larger class of weight-based 
partitions, and in this sense Lemma \ref{lem:HWDF} is strictly more general.
\end{remark}

\subsection{Partition graph of the support partition}

In this subsection, we consider the following partition of $\mathbb{F}_q^k$ and refer to it as \emph{support partition},
$$\mathcal{S}=\{S_A : A \subseteq [k]\}, \quad \text{where } S_A=\{u \in \mathbb{F}_q^k : \mathrm{supp}(u)=A\}\ \text{for all } A\subseteq [k].$$

The support partition is strictly finer than the weight partition, since each weight class 
$W_i = \{x \in \mathbb{F}_q^k: \mathrm{wt}(x)=i\}$ is subdivided into $\binom{k}{i}$ blocks $S_A$ of vectors having the same support. There are many functions whose domain partition is exactly $\mathcal{S}$, and an $(\mathcal{S},t)$-encoding automatically works as an $(f,t)$-FCC for all such functions. Two simple examples are given below.
\begin{enumerate}
    \item The support function $f(u)=\mathrm{Supp}(u)$ for all $u\in\mathbb{F}_q^k$.
    \item The function $f:\mathbb{F}_q^k \to \{0,1,\ldots,2^k\}$ defined by
    \[
        f(u)=\sum_{i=1}^k u_i^{q-1}\,2^{i-1}, 
    \]
    for all $u=(u_1, u_2, \ldots, u_k)\in \mathbb{F}_q^k$. Since $g(a)=1$ for all $a\neq 0$ and $g(0)=0$, the value $f(u)$ encodes the 
    support of $u$ in binary, and therefore its domain partition is $\mathcal{S}$.
\end{enumerate}

The next lemma establishes the existence of a full-size clique in the partition graph of the support partition of $\mathbb{F}_q^k$.

\begin{lemma}\label{lem:SP}
    The partition graph of the support partition $\mathcal{S}=\{S_A : A \subseteq [k]\}$ contains a full size clique of size $2^k$.
\end{lemma}

\begin{proof}
    Let $x\in S_A$ and $y\in S_B$ for some different subsets $A$ and $B$ of $[k]$. That means $\mathrm{supp}(x)=A$ and $\mathrm{supp}(y)=B$. Then 
    $$d(x,y)=|A\Delta B|+\left|\{i\in A \cap B : x_i \neq y_i\}\right| \geq |A \Delta B|,$$
    where $\Delta$ denotes the symmetric difference of two sets.
    Therefore, we have $d(S_A, S_B) \geq |A\Delta B|$.
    On the other hand, fix a non zero symbol $a\in \mathbb{F}_q$, and define a vector $u_A=\left((u_A)_1, (u_A)_2, \ldots,(u_A)_k\right)$ for each $A\subseteq[k]$ such that
     $$(u_A)_i=\begin{cases}
        0 & \text{if } i \notin A \\
        a & \text{if } i \in A
    \end{cases}.$$
    Then $\mathrm{supp}(u_A)=A$, so $u_A \in S_A$. Furthermore, for $A,B\subseteq[k], A \neq B$, we have $d(u_A, u_B)=|A \Delta B|$. Combining with the previous inequality, we obtain
$$d(S_A, S_B)=d(u_A, u_B)= |A\Delta B| \quad \text{for all } A\neq B.$$
By the definition of partition graph, there is an edge between $u_A$ and $u_B$ for all $A, B\subseteq[k], A\neq B$, and the set $\mathcal{C}=\{u_A: A\subseteq[k]\}$ with $|\mathcal{C}|=2^k$ forms a clique in the partition graph $\mathcal{G}_\mathcal{S}$.
\end{proof}

In the next example, we consider the support partition of $\mathbb{F}_3^3$ and generate a clique for it.

\begin{example}\label{ex10}
Consider the space $\mathbb{F}_3^3$.
For each subset $A\subseteq\{1,2,3\}$, define the block
$
S_A=\{x\in\mathbb{F}_3^3 : \mathrm{supp}(x)=A\}.
$
Then the support partition is $
\mathcal{S}=\{S_A : A\subseteq\{1,2,3\}\}$ consisting of $2^3=8$ blocks:
\begin{align*}
S_{\emptyset} &= \{(0,0,0)\},\\
S_{\{1\}} &= \{(1,0,0),(2,0,0)\},\\
S_{\{2\}} &= \{(0,1,0),(0,2,0)\},\\
S_{\{3\}} &= \{(0,0,1),(0,0,2)\},\\
S_{\{1,2\}} &= \{(1,1,0), (2,2,0), (1,2,0), (2,1,0)\},\\
S_{\{1,3\}} &= \{(1,0,1), (2,0,2), (1,0,2), (2,0,1)\},\\
S_{\{2,3\}} &= \{(0,1,1), (0,2,2), (0,1,2), (0,2,1)\},\\
S_{\{1,2,3\}} &= \{(1,1,1), (2,2,2), (1,1,2), (1,2,1), (2,1,1), (1,2,2), (2,1,2), (2,2,1)\}.
\end{align*}
Using the representative vectors $u_A$ defined in Lemma~\ref{lem:SP} with $a=1$, we obtain the clique 
$$
\mathcal{C}= \{(0,0,0), (1,0,0), (0,1,0), (0,0,1),(1,1,0), (1,0,1), (0,1,1), (1,1,1)\},
$$
shown in Fig. \ref{fig:ex10}.

\begin{figure}[h]
\centering
\begin{tikzpicture}[scale=1.5]

    \def\r{2}

    \coordinate (v1) at (90:\r);      
    \coordinate (v2) at (45:\r);      
    \coordinate (v3) at (0:\r);       
    \coordinate (v4) at (315:\r);     
    \coordinate (v5) at (270:\r);     
    \coordinate (v6) at (225:\r);     
    \coordinate (v7) at (180:\r);     
    \coordinate (v8) at (135:\r);     

    \foreach \i in {1,...,8} {
        \pgfmathtruncatemacro{\nexti}{\i+1} \foreach \j in {\nexti,...,8} {
            \draw[thin] (v\i)--(v\j);
        }
    }
    \foreach \i in {1,...,8} {
        \fill (v\i) circle (1.5pt);
    }

    \node[above]          at (v1) {$(0,0,0)$};
    \node[above right]    at (v2) {$(1,0,0)$};
    \node[right]          at (v3) {$(0,1,0))$};
    \node[below right]    at (v4) {$(0,0,1))$};
    \node[below]          at (v5) {$(1,1,0)$};
    \node[below left]     at (v6) {$(1,0,1)$};
    \node[left]           at (v7) {$(0,1,1)$};
    \node[above left]     at (v8) {$(1,1,1)$};

\end{tikzpicture}
\caption{Representation of the clique $\mathcal{C}$ from Example \ref{ex10}.}
\label{fig:ex10}
\end{figure}
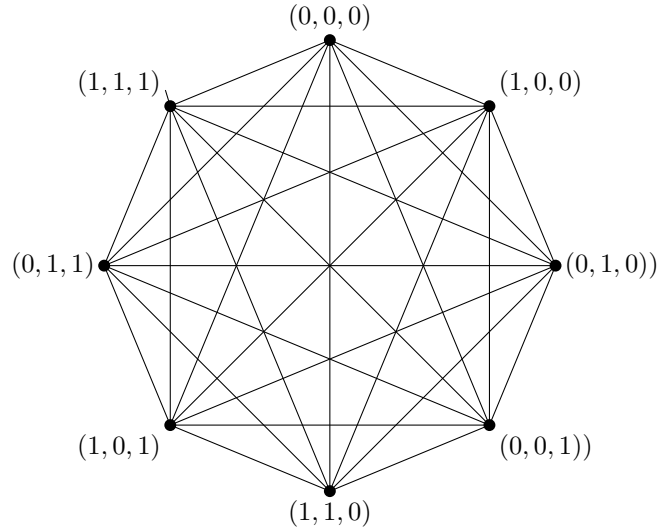

\end{example}

\begin{example}\label{ex:F4support}
Consider the space $\mathbb{F}_4^2$, where 
$\mathbb{F}_4 = \{0,1,\omega,\omega^2\}$
with $\omega$ satisfying $\omega^2 = \omega + 1$.
For each subset $A \subseteq \{1,2\}$, define the block
$S_A = \{x \in \mathbb{F}_4^2 : \mathrm{supp}(x)=A\}.$
Then the support partition $\mathcal{S}=\{S_A : A\subseteq\{1,2\}\}$
consists of $2^2=4$ blocks:
\begin{align*}
S_{\emptyset} 
&= \{(0,0)\},\\[2pt]
S_{\{1\}} 
&= \{(1,0), (\omega,0), (\omega^2,0)\},\\[2pt]
S_{\{2\}} 
&= \{(0,1), (0,\omega), (0,\omega^2)\},\\[2pt]
S_{\{1,2\}} 
&= \{(1,1),\ (1,\omega),\ (1,\omega^2), (\omega,1),\ (\omega,\omega),\ (\omega,\omega^2), (\omega^2,1),\ (\omega^2,\omega),\ (\omega^2,\omega^2)\}.
\end{align*}

Now, following the construction in Lemma~\ref{lem:SP} with $a=1\in \mathbb{F}_4$, we get a clique
$$
\mathcal{C} = \{u_A : A\subseteq\{1,2\}\}
= \{(0,0), (1,0), (0,1), (1,1)\}
$$
of size $2^2=4$ in the partition graph $G_{\mathcal{S}}$ of the support partition of $\mathbb{F}_4^2$.
\end{example}

Combining Lemma~\ref{lem:SP} with Theorem~\ref{thm:clique}, we obtain the following result.

\begin{corollary}\label{col:SP1}
For any positive integers $k$ and $t$, and for a fixed $a\in\mathbb{F}_q$, 
define a vector $u_A=\left((u_A)_1,\ldots,(u_A)_k\right)$ by
$(u_A)_i = a$ if $i\in A$ and $(u_A)_i = 0$ otherwise, for each $A\subseteq [k]$.
Then
$$
r_{\mathcal{S}}(k,t)
    = N\big(D_{\mathcal{S}}(t \,;\, u_A : A\subseteq [k])\big).
$$
\end{corollary}

In the following corollary, we compute a lower bound on the redundancy of a 
$(\mathcal{S},t)$-encoding over $\mathbb{F}_q$ using the Plotkin bound given in 
Lemma~\ref{lem:RRHH}, and a simple upper bound using an ECC.

\begin{corollary}\label{col:SP2}
    For any positive integers $k$ and $t$, the redundancy of a 
    $(\mathcal{S},t)$-encoding over $\mathbb{F}_q$ satisfies
    $$
        \frac{2^k q}{4^k (q-1) - a(q-a)} 
        \sum_{s=1}^{\min (k, 2t)} (2t+1-s) \binom{k}{s} 
        \leq r_{\mathcal{S}}(k,t)
        \leq N_q(2^k,2t),
    $$
    where $a = 2^k \bmod q$, and $N_q(M,d)$ denotes the minimum length of a 
    $q$-ary ECC with $M$ codewords and minimum distance $d$.
\end{corollary}

\begin{proof}
From Corollary~\ref{col:SP1}, we know that $ r_{\mathcal{S}}(k,t)
    = N\left(\mathcal{D}_{\mathcal{S}}(t;u_A : A\subseteq [k])\right),$
so we need to bound 
$N\left(\mathcal{D}_{\mathcal{S}}(t;u_A : A\subseteq [k])\right)$.
Let $U=\{u_A : A \subseteq [k]\}$ which is equivalent to the space $\mathbb{F}_2^k$. If the rows and columns of the $2^k \times 2^k$ PDRM 
$\mathcal{D}_{\mathcal{S}}(t;u_A : A\subseteq [k])$ are indexed by the subsets of $[k]$, then the $(A,B)$-th entry of this PDRM is
$$
    \mathcal{D}_{A,B}=\max(2t+1-d(u_A, u_B),0), \quad A\neq B.
$$
Therefore, the sum of the upper triangular part of this matrix is
$$
    \Gamma(k,t)
    = \sum_{s=1}^k N_s \max(2t+1-s,0)
    = \sum_{s=1}^{\min(k,2t)} N_s (2t+1-s),
$$
where $N_s$ denotes the number of unordered pairs $(u_A, u_B)$ in $U$ at 
distance~$s$.

For any $s\in [k]$, the number of ordered pairs $(u_A, u_B)$ in $U$ at 
distance $s$ is $2^k \binom{k}{s}$, as $u_A$ can be chosen in $2^k$ ways and, 
by flipping any $s$ coordinates of $u_A$, we obtain a vector $u_B$ at distance 
$s$. Therefore, the number of unordered pairs $(u_A, u_B)$ in $U$ at distance $s$ 
is $N_s = 2^{k-1} \binom{k}{s}$. Hence,
$$
    \Gamma(k,t)
    = \sum_{s=1}^{\min(k,2t)} 2^{k-1} (2t+1-s)\binom{k}{s}.
$$
Finally, applying the Plotkin bound from Lemma~\ref{lem:RRHH}, we obtain
$$
N\left(\mathcal{D}_{\mathcal{S}}(t;u_A:A\subseteq[k])\right)\;\ge\;
\frac{2^k q}{4^k (q-1) - a(q-a)}
\sum_{s=1}^{\min(k,2t)} (2t+1-s)\binom{k}{s},
$$
where $a = 2^k \bmod q$. Furthermore, since by the definition of the PDRM,
$$
\mathcal{D}_{A,B} = \max(2t+1-d(u_A, u_B),0) \le 2t,
$$
a $q$-ary ECC of length $N_q(2^k,2t)$, size $2^k$, and minimum 
distance $2t$ is a $\mathcal{D}_{\mathcal{S}}(t;u_A: A\subseteq[k])$-code. 
Therefore, we obtain the claimed upper bound.
\end{proof}

For the parameters in Example~\ref{ex10}, namely $q=3$, $k=3$, and $t=2$, 
Corollary~\ref{col:SP2} gives the lower bound 
$r_{\mathcal{S}}(k,t) \ge 4.38$, and hence $r_{\mathcal{S}}(k,t)\ge 5$.  
A $D$-code of length $6$ for the PDRM corresponding to the clique $\mathcal{C}$ 
is given by
\begin{align*}
z_{\emptyset}   &= (0,0,0,0,0,0), &
z_{\{1\}}       &= (0,0,1,1,1,1),\\
z_{\{2\}}       &= (0,0,1,2,2,2), &
z_{\{3\}}       &= (0,1,0,1,1,2),\\
z_{\{1,2\}}     &= (0,1,0,0,2,1), &
z_{\{1,3\}}     &= (0,0,2,0,2,2),\\
z_{\{2,3\}}     &= (0,0,2,1,0,1), &
z_{\{1,2,3\}}   &= (0,0,0,2,1,0).
\end{align*}
Thus, for this example, the achievable redundancy is only one symbol larger than the lower bound.

\section{Block-preserving contraction of partitions}\label{sec:BPC}


The existence of a full-size clique in the partition graph guaranties a maximal reduction in the size of the problem of constructing an FCPC with optimal length from the entire space $\mathbb{F}_q^k$ to a set of size $E$, where $E$ denotes the number of blocks in the partition. However, such full-size cliques need not exist for an arbitrary partition, as demonstrated in Lemma~\ref{lem:HWDF} for certain
consecutive grouped weight partitions, and even when they do, identifying them may be computationally difficult. In these situations, one may seek an intermediate solution that, while not achieving the minimal reduction to $E$, still reduces the search space to a subset of $\mathbb{F}_q^k$ that is strictly
smaller than $q^k$ and larger than $E$.

In this section, we introduce the notion of a \emph{block-preserving contraction},
which maps the entire space $\mathbb{F}_q^k$ onto a smaller subset $U$ while
preserving block membership and non-increasing inter-block distances. This notion
naturally generalizes the full-size clique condition, since a full-size clique
corresponds precisely to a block-preserving contraction of minimum possible
size~$E$. We show that whenever a block-preserving contraction exists for a
partition $\mathcal{P}$, the optimal redundancy $r_{\mathcal{P}}(k,t)$ can be
determined solely from the vectors in $U$ together with the induced distance
constraints.

\begin{definition}[Block-preserving contraction]
Let $\mathcal{P} = \{P_1,\ldots,P_E\}$ be a partition of $\mathbb{F}_q^k$.
We say that a subset $U \subseteq \mathbb{F}_q^k$ together with a map
$\phi:\mathbb{F}_q^k \to U$ forms a block-preserving contraction of $\mathcal{P}$ if
\begin{itemize}
    \item $\phi(u) \in P_i$ whenever $u \in P_i$, and
    \item for all $u,v$ belonging to different blocks of $\mathcal{P}$,
    $$
        d(\phi(u),\phi(v)) \le d(u,v).
    $$
\end{itemize}
In this case, we refer to the pair $(U,\phi)$ as a block-preserving contraction of $\mathcal{P}$.
\end{definition}

A trivial block-preserving contraction of any partition of $\mathbb{F}_q^k$ is $(U=\mathbb{F}_q^k, \phi=I)$, where $I$ denotes the identity map on $\mathbb{F}_q^k$. Also, for any block-preserving contraction $(U, \phi)$ of a partition $\mathcal{P}=\{P_1, P_2, \ldots, P_E\}$, we have
$$E\leq |U| \leq q^k.$$ Now we give an example for block-preserving contraction as follows.

\begin{example}\label{ex12}
Consider the partition $\mathcal{P} = \{P_1,P_2,P_3\}$ of $\mathbb{F}_2^4$ defined by
\begin{align*}
P_1 &= \{0000,\,0001,\,0010,\,0100\},\\[1mm]
P_2 &= \{0011,\,0101,\,0110,\,1001,\,1010,\,1100,\,0111,\,1011,\,1101,\,1000,\,1110\},\\[1mm]
P_3 &= \{1111\}.
\end{align*}
Here we have $d(P_1,P_2)=1$, $d(P_2,P_3)=1$, and $d(P_1,P_3)=3$. For the existence
of a full-size clique, there must exist vectors $x_1\in P_1$, $x_2\in P_2$, and
$x_3\in P_3$ that simultaneously realize these distances. However, by the triangle
inequality,
$$
d(x_1,x_3) \le d(x_1,x_2) + d(x_2,x_3) = 1 + 1 = 2,
$$
which makes it impossible to have $d(x_1,x_3)=3$. Therefore, no full-size clique
exists in the partition graph of the partition $\mathcal{P}$.

Define the set $U = \{0001,\,0101,\,1101,\,1111\},$
and define $\phi:\mathbb{F}_2^4 \to U$ by
$$
\phi(x)=
\begin{cases}
0001, & x\in P_1,\\[1mm]
0101, & x\in \{0011,0101,0110,1001,1010,1100,1000\},\\[1mm]
1101, & x\in \{0111,1011,1101,1110\},\\[1mm]
1111, & x = 1111.
\end{cases}
$$
Then $(U,\phi)$ is a block-preserving contraction of $\mathcal{P}$.

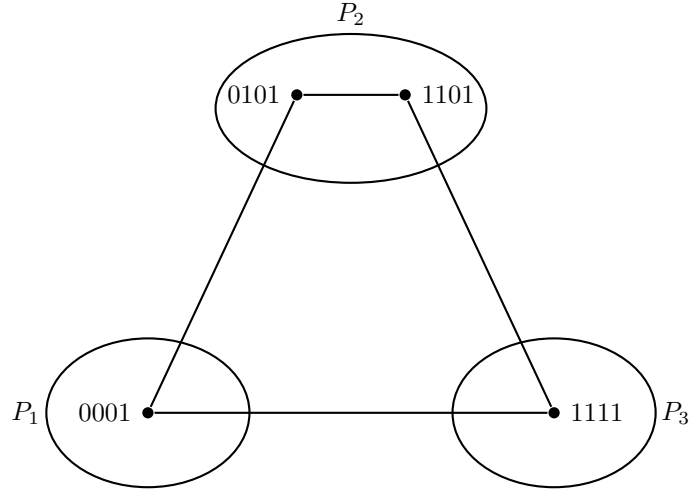
\begin{figure}[!h]
\centering
\begin{tikzpicture}[thick,scale=0.9]


\draw (0,3) ellipse (2.0 and 1.1);
\node at (0,4.4) {$P_2$};

\draw (-3,-1.5) ellipse (1.5 and 1.1);
\node at (-4.8,-1.5) {$P_1$};

\draw (3,-1.5) ellipse (1.5 and 1.1);
\node at (4.8,-1.5) {$P_3$};


\node[circle,fill=black,inner sep=1.5pt,label=left:{$0001$}] (u1) at (-3,-1.5) {};

\node[circle,fill=black,inner sep=1.5pt,label=left:{$0101$}] (u2) at (-0.8,3.2) {};
\node[circle,fill=black,inner sep=1.5pt,label=right:{$1101$}] (u3) at (0.8,3.2) {};

\node[circle,fill=black,inner sep=1.5pt,label=right:{$1111$}] (u4) at (3,-1.5) {};


\draw (u1) -- (u2);
\draw (u2) -- (u3);
\draw (u3) -- (u4);

\draw (u1) -- (u4);

\end{tikzpicture}
\caption{Representation of block-preserving contraction $(U,\phi)$ given in Example~\ref{ex12}.}
\label{fig:ex12}
\end{figure}

\end{example}

The following theorem shows that a block-preserving contraction is sufficient to represent the optimal redundancy.

\begin{theorem}\label{thm:BPC}
Let $\mathcal{P}$ be a partition of $\mathbb{F}_q^k$ and suppose that there exists a block-preserving contraction $(U,\phi)$ of $\mathcal{P}$.
Then
$$
    r_{\mathcal{P}}(k,t)
    = N\big(\mathcal{D}_{\mathcal{P}}(t, u : u \in U)\big).
$$
\end{theorem}

\begin{proof}
Let $N\big(\mathcal{D}_{\mathcal{P}}(t, u : u \in U)\big)=r'$. From Theorem \ref{thm:LUB_partition}, we have 
$r' \le r_{\mathcal{P}}(k,t).$
Now we construct a $(\mathcal{P},t)$-encoding with redundancy $r'$ to prove the reverse inequality.
By the definition of $N\big(\mathcal{D}_{\mathcal{P}}(t, u : u \in U)\big)$, there exists a $\mathcal{D}_{\mathcal{P}}(t, u : u \in U)$-code $\{z_u : u \in U\}$ with redundancy $r'$.
We use the block-preserving contraction $(U,\phi)$ to obtain an encoding $\mathcal{C} : \mathbb{F}_q^k \rightarrow \mathbb{F}_q^{k+r'}$ defined as 
$$
\mathcal{C}(x) = \bigl(x,\, z_{\phi(x)}\bigr) \quad \text{for all } x \in \mathbb{F}_q^k.
$$
Thus, each $x$ is assigned the same redundancy vector as its contracted image $\phi(x)$.

We claim that $\mathcal{C}$ is a $(\mathcal{P},t)$-encoding.
Let $x,y \in \mathbb{F}_q^k$ belong to different blocks of $\mathcal{P}$, say $x \in P_i$, $y \in P_j$ with $i \ne j$.
Since $(U,\phi)$ is a block-preserving contraction, we have $\phi(x) \in P_i, \ \phi(y) \in P_j,$ and 
$$d(\phi(x), \phi(y)) \le d(x,y).$$
So $\phi(x)$ and $\phi(y)$ also lie in different blocks of $\mathcal{P}$, and both lie in $U$. By the definition of a $\mathcal{D}$-code, we have 
$$d(z_{\phi(x)}, z_{\phi(y)}) \geq 2t+1-d(\phi(x), \phi(y)) \geq 2t+1- d(x, y).$$
This means  
$$d(\mathcal{C}(x),\mathcal{C}(y))=d(x,y)+d(z_{\phi(x)}, z_{\phi(y)})\geq 2t+1.$$
This holds for all $x,y$ in different blocks of $\mathcal{P}$, so $\mathcal{C}$ is a $(\mathcal{P},t)$-encoding with redundancy $r'$.
Thus
$$
r_{\mathcal{P}}(k,t) \le  r' = N\big(\mathcal{D}_{\mathcal{P}}(t, u : u \in U)\big),
$$
and we have $r_{\mathcal{P}}(k,t)  = N\big(\mathcal{D}_{\mathcal{P}}(t, u : u \in U)\big).$
\end{proof}

In Example \ref{ex12}, it is not possible to reduce the problem size to $E=3$
vectors corresponding to a full-size clique in the partition graph. However, the
block-preserving contraction $(U,\phi)$ still allows a reduction of the problem size
from the entire space of $16$ vectors to a set of $4$ vectors, which serves as the next best
alternative when a full-size clique does not exist.

\begin{remark}
     The notion of a block-preserving contraction generalizes the full-size clique 
condition of Theorem~\ref{thm:clique} (equivalently, the set of representatives 
condition for FCCs in Corollary~\ref{col1}). 
In particular, if a partition admits a block-preserving contraction $(U,\phi)$ 
with $|U| = E$, then the set $U$ forms a full-size clique. 
In other words, any full-size clique 
$\mathcal{C} = \{u_1, u_2, \ldots, u_E\}$ is a block-preserving contraction of the minimum possible size, with the associated map $\phi_{\mathcal{C}}$ defined by $\phi_{\mathcal{C}}(x) = u_i$ whenever $x$ lies in the same block of the partition as $u_i$.
\end{remark}

The following lemma clarifies how block-preserving contractions transfer from a finer partition to a coarser one.

\begin{lemma}\label{lem:BPC1}
Suppose $\mathcal{P}$ and $\mathcal{Q}$ are two partitions of $\mathbb{F}_q^k$ and 
$\mathcal{P}$ is a refinement of $\mathcal{Q}$.  
Then any block-preserving contraction $(U,\phi)$ of $\mathcal{P}$ is also a block-preserving contraction of $\mathcal{Q}$.
\end{lemma}

\begin{proof}
Let $(U,\phi)$ be a block-preserving contraction of $\mathcal{P}=\{P_1,\ldots,P_m\}$.  
We need to show that it is also a block-preserving contraction of $\mathcal{Q}=\{Q_1,\ldots,Q_s\}$.

\noindent
{\textbf{(1)}} Let $u \in Q_i$ for some $i \in [s]$.  
Since $\mathcal{P}$ is a refinement of $\mathcal{Q}$, 
there exists $P_r$ such that $u \in P_r \subseteq Q_i$.  
By the definition of block-preserving contraction, we have
$$
\phi(u) \in P_r \subseteq Q_i.
$$

\noindent
{\textbf{(2)}} Let $u,v \in \mathbb{F}_q^k$ be from different blocks of $\mathcal{Q}$.  
Say $u \in Q_i$ and $v \in Q_j$ with $i \ne j$.  
Since $\mathcal{P}$ is a refinement of $\mathcal{Q}$, 
there exist $P_r$ and $P_s$ such that
$$
u \in P_r \subseteq Q_i, \qquad v \in P_s \subseteq Q_j.
$$
As $(U,\phi)$ is a block-preserving contraction of $\mathcal{P}$, we have
$
d(\phi(u),\phi(v)) \le d(u,v).
$

Therefore, $(U,\phi)$ is also a block-preserving contraction of $\mathcal{Q}$.
\end{proof}

As we have seen, any block-preserving contraction $(U,\phi)$ of a partition 
$\mathcal{P}$ is also a block-preserving contraction of a coarser partition 
$\mathcal{Q}$. However, $\mathcal{Q}$ can have another block-preserving 
contraction $(U',\phi')$ such that $|U'| < |U|$. 
For example, the support partition $\mathcal{S}$ is a refinement of the weight 
partition $\mathcal{W}$, and $\mathcal{S}$ has a clique
\[
U = \{u_A : A \subseteq [k]\}, \qquad 
(u_A)_i = 0 \text{ if } i \notin A,\; (u_A)_i = 1 \text{ if } i \in A.
\]
By the Lemma \ref{lem:BPC1}, $(U,\phi_U)$ is a block-preserving contraction of $\mathcal{W}$, 
where $\phi_{U}(x) = u_A$ whenever $x$ lies in the same block of $\mathcal{S}$ 
as $u_A$.
However, $\mathcal{W}$ (a coarser partition) has a clique of its own defined as
\[
U' = \{v_0,v_1,\ldots,v_k\},
\qquad
v_i = (\underbrace{1,1,\ldots,1}_{i\ \text{times}},0,0,\ldots,0),
\]
such that $|U'| = k+1 < |U| = 2^k$.  
For our purpose, we want a block-preserving contraction of size as small as 
possible.

The Lemma~\ref{lem:BPC1} leads to the following general statements.

\begin{corollary}\label{col:BPC1}
Using Lemma~\ref{lem:BPC1}, we have the following results.
\begin{enumerate}
    \item 
    The domain partition of any function 
    $f:\mathbb{F}_q^k \to \mathrm{Im}(f)$ of the form 
    $f(x)=g(\mathrm{supp}(x))$ has $(U,\phi_U)$ as a block-preserving contraction, 
    where
    \[
    U = \{u_A : A \subseteq [k]\}, \qquad 
    (u_A)_i = 0 \text{ if } i \notin A,\; (u_A)_i = 1 \text{ if } i \in A,
    \]
    and $\phi_U(x)=u_A$ whenever $x$ lies in the same block of $\mathcal{S}$ as $u_A$.

    \item 
    Let $\mathcal{G}=\{G_1,G_2,\ldots,G_m\}$ be a grouped weight partition of 
    $\mathbb{F}_q^k$. Since the weight partition is finer than any grouped weight partition, $(U',\phi_{U'})$ is a block-preserving contraction of $\mathcal{G}$, where 
    $$
    U'=\{v_i : i=0,1,\ldots,k\}, \qquad v_i = (\underbrace{1,1,\ldots,1}_{i\ \text{times}},0,0,\ldots,0),
    $$
   and $\phi_{U'}(x)=v_{\mathrm{wt}(x)}$ for all $x\in \mathbb{F}_q^k$.

    \item 
    The domain partition of any function 
    $f:\mathbb{F}_q^k \to \mathrm{Im}(f)$ of the form 
    $f(x)=g(\mathrm{wt}(x))$ has $(U',\phi_{U'})$ as a block-preserving contraction.
\end{enumerate}
\end{corollary}

The following corollary holds directly from Theorem \ref{thm:BPC} and Corollary \ref{col:BPC1}.

\begin{corollary}\label{col:BPC_GWP}
    For any grouped weight partition $\mathcal{G}$ of $\mathbb{F}_q^k$, 
    $$r_{\mathcal{G}}(k,t)=N_q(\mathcal{D}_{\mathcal{G}}(t; u_0,u_1,\ldots, u_k)),$$
    where $u_i = (\underbrace{1,1,\ldots,1}_{i\ \text{times}},0,0\ldots,0)$ for all $i\in \{0,1,\ldots, k\}$.
\end{corollary}

Since the domain partition of a Hamming weight distribution function $\Delta_T$ is a grouped weight partition, Lemma 3.1 of \cite{GXZZ2025}, which provides the redundancy expression for $\Delta_T$, appears as a special case of the above corollary. 
More generally, Corollary \ref{col:BPC_GWP} applies to any function that depends only on the Hamming weight of its input, i.e., to any function of the form $f(u)=g(wt(u)), \ u \in \mathbb{F}_q^k$.

As shown in Lemma~\ref{lem:HWDF}, for certain consecutive grouped weight
partitions, a full-size clique does not exist. However, from the second part of
Corollary~\ref{col:BPC1}, such partitions still admit a block-preserving
contraction of size $k+1$.


Now, we give a condition for a coset partition so that it has a block-preserving contraction.
For any subset $J \subseteq [k]$, define
$$
\mathfrak{S}_J
= \{u \in \mathbb{F}_q^k : \operatorname{supp}(u) \subseteq J\}=\bigcup_{A\subseteq J} S_A,
$$
where $S_A$ denotes a block of support partition that contains all vectors of $\mathbb{F}_q^k$ with support $A$. 

\begin{lemma} \label{lem:BPC2}
Let $V \le \mathbb{F}_q^k$ be a subspace. Suppose there exists a subset
$J \subseteq [k]$ such that $\mathfrak{S}_{[k]\setminus J} \subseteq V$.
Then $(\mathfrak{S}_J,\phi_J)$ is a block-preserving contraction of the coset partition of $V$, where
$\phi_J : \mathbb{F}_q^k \to \mathfrak{S}_J$ is defined as follows: for all
$u=(u_1,u_2,\ldots,u_k) \in \mathbb{F}_q^k$,
$$
(\phi_J(u))_i =
\begin{cases}
u_i, & \text{if } i \in J,\\
0, & \text{if } i \notin J,
\end{cases}
\qquad \text{for } i \in [k].
$$
\end{lemma}

\begin{proof}
    Let the coset partition of $V$ be denoted by $\mathcal{V}$. 
    \begin{itemize}
        \item For any $u\in \mathbb{F}_q^k$, we have $\mathrm{supp}(u-\phi_J(u)) \subseteq [k]\setminus J$, which means $u-\phi_J(u) \in \mathfrak{S}_{[k]\setminus J} \subseteq V$. Therefore, $u$ and $\phi_J(u)$ belong to the same coset of $V$.
    \item  For any two vectors $u,v \in \mathbb{F}_q^k$, we have
$$d(\phi_J(u), \phi_J(v))\leq d(u,v),$$
as $\phi_J$ only deletes coordinates from a fixed set $[k]\setminus J$, and does not  create new disagreements.
    \end{itemize}

    Therefore, $(\mathfrak{S}_J, \phi_J)$ is a block-preserving contraction of $\mathcal{V}$.
\end{proof}

In particular, if $\dim(V)=\ell$ and there exists a subset $J\subseteq[k]$ with
$|J|=\ell$ satisfying $\mathfrak{S}_{[k]\setminus J}\subseteq V$, then the
block-preserving contraction $(\mathfrak{S}_J,\phi_J)$ yields a clique of full
size, since $|\mathfrak{S}_J|=q^{\ell}$ equals the number of cosets in the coset
partition $\mathcal{V}$.

\begin{example}
Consider the space $\mathbb{F}_2^5$ and the subspace
$$
V =\{(a,a,a,b,c): a,b,c\in\mathbb{F}_2\}\;\le\;\mathbb{F}_2^5 .
$$
Then $\dim(V)=3$, and the coset partition $\mathcal{V}$ of $V$ has
$4$ cosets. For instance, taking representatives $00000$, $10000$, $01000$, $00100$, the four cosets are
$$
\renewcommand{\arraystretch}{1.05}
\begin{array}{c|c|c|c}
V & 10000+V & 01000+V & 00100+V\\ \hline
00000 & 10000 & 01000 & 00100\\
00001 & 10001 & 01001 & 00101\\
00010 & 10010 & 01010 & 00110\\
00011 & 10011 & 01011 & 00111\\
11100 & 01100 & 10100 & 11000\\
11101 & 01101 & 10101 & 11001\\
11110 & 01110 & 10110 & 11010\\
11111 & 01111 & 10111 & 11011
\end{array}
$$
It can be verified that for any two distinct cosets of $\mathcal{V}$ there exist
vectors at Hamming distance~$1$. Consequently, the existence of a clique of full
size $|\mathcal{V}|=4$ would require four vectors in $\mathbb{F}_2^5$ that are
pairwise at Hamming distance~$1$, which is impossible.

Now take $J=\{1,2,3\}$.
Then $[k]\setminus J =\{4,5\}$, and
$$
\mathfrak{S}_{[k]\setminus J}
=\{x\in\mathbb{F}_2^5 : x_J=0\}
=\{(0,0,0,b,c): b,c\in\mathbb{F}_2\}
\subseteq W,
$$
Hence, by Lemma~\ref{lem:BPC2},
$(\mathfrak{S}_J,\phi_J)$ is a block-preserving contraction of $\mathcal{V}$, where $(\phi_J(u))_i = u_i$ for $i \in J$ and $(\phi_J(u))_i = 0$ for $i \in [k]\setminus J$.
Moreover,
$$
\mathfrak{S}_J
=\{x\in\mathbb{F}_2^5 : x_{4}=x_{5}=0\}
=\{00000,10000,01000,00100,11000,10100,01100,11100\},
$$
so $|\mathfrak{S}_J|=2^{|J|}=8$. In particular,
\[
|\mathcal{V}|=4 \;\le\; |\mathfrak{S}_J|=8 \;\le\; 2^5.
\]
\end{example}

\section{Function Class Privacy in FCPCs}
\label{sec:privacy}

In this section, we discuss a conceptual advantage of FCPCs over FCCs, namely a form of \emph{function class privacy}
with respect to the function being computed. Here, a \emph{function class} refers to the collection of all functions on the domain $\mathbb{F}_q^k$ that induce the same domain partition. This notion of function class privacy arises naturally from the fact that an FCPC depends only on the domain partition induced by the function, and not on the function itself.

\subsection{Partial privacy via partition-based encoding}

Recall that any function $f:\mathbb{F}_q^k \to S$ induces a domain partition
$\mathcal{P}_f$, where two vectors belong to the same block if and only if
they are mapped to the same function value. This partition naturally defines
a \emph{function class}, consisting of all functions on $\mathbb{F}_q^k$ that
induce the same domain partition $\mathcal{P}_f$.
In FCCs, the transmitter knows the exact function $f$, and hence no ambiguity regarding the function is present. In contrast, an FCPC requires the transmitter to know only the induced partition $\mathcal{P}_f$, and therefore only the corresponding function class is revealed not the exact function.
In this sense, FCPCs provide 
\emph{function class privacy}, i.e., the transmitter cannot distinguish between
functions belonging to the same function class.

If the codomain is unrestricted, infinitely many functions belong to the
function class corresponding to a given partition. Even when the codomain is
known and its size $|S|=H$ is fixed, a partition with $E$ blocks can correspond
to many distinct functions, obtained by assigning elements of $S$ to the blocks
in different ways. As $H \ge E$, the size of the function class corresponding to
a fixed partition with $E$ blocks is
$$\binom{H}{E} E! = \frac{H!}{(H-E)!}.$$ 
Consequently, partitions with a larger number of blocks admit a larger function
class, leading to a higher degree of function ambiguity. The following example
illustrates this.

\begin{example}
\label{ex13}
Consider the partition $\mathcal{P}$ of $\mathbb{F}_2^4$ defined as
\begin{align*}
\mathcal{P}=\{&
P_1=\{0000,0001,0110,0111\},\;
P_2=\{0010,0011,0100,0101\},\\
&
P_3=\{1000,1001,1110,1111\},\;
P_4=\{1100,1101,1010,1011\}\}.
\end{align*}
This partition is the same as the one considered in Example~\ref{ex2}. For a
fixed codomain $\mathbb{F}_2^2=\{00,01,10,11\}$, there are $4!=24$ distinct functions whose
 domain partition is exactly $\mathcal{P}$. Among these functions, some
are linear while others are nonlinear.

Three linear functions inducing the partition $\mathcal{P}$ are given in
Example~\ref{ex2}. As an illustration of function ambiguity, we now give a
nonlinear function $f:\mathbb{F}_2^4 \to S$ that induces the same partition,
defined by assigning
\[
f(x)=
\begin{cases}
01, & x\in P_1,\\
00, & x\in P_2,\\
11, & x\in P_3,\\
10, & x\in P_4.
\end{cases}
\]
Although this function is nonlinear, it induces the same domain partition
$\mathcal{P}$ as the linear functions in Example~\ref{ex2}. From the
transmitter’s perspective, these functions are indistinguishable when only the
partition $\mathcal{P}$ is revealed.
\end{example}

\subsection{Function class privacy and special functions}

For certain functions, such as locally binary functions, additional structure can be exploited to obtain particularly simple or efficient code constructions. In such cases, the domain partition may possess properties that enable very small $\mathcal{D}$-codes. As a construction for $2t$-locally binary functions is given in \cite{LBWY2023} achieving the lower bound $2t$, which only
require two vectors in the corresponding $\mathcal{D}$-code.

When only the partition is revealed, the transmitter may not be able to determine whether the partition arises from a special function, especially when the domain $\mathbb{F}_q^k$ is large. However, this
does not affect the existence or correctness of an optimal FCPC for a given partition. If desired, the receiver may reveal limited additional structural information about the function along with the partition. This information does not disclose the exact function, but enables the transmitter to employ constructions tailored to that special function.

It is possible for two distinct functions to induce the same domain partition and to exhibit the same structural characteristics, while differing in their precise mappings. From the transmitter's perspective, such functions are indistinguishable when only the partition and possibly limited structural information are revealed. The same FCPC applies to both functions,
and decoding remains correct for the receiver. Here is one such example.

\begin{example}
Consider a setting in which the receiver reveals to the transmitter a partition
$\mathcal{P}$ of $\mathbb{F}_q^{15}$ together with the information that the
desired function is a locally $(2t,2)$-bounded function. Equivalently, the receiver specifies that the partition $\mathcal{P}$ is a locally $(2t,2)$-bounded partition for $t=1$.
The partition is given by
\begin{multline*}
\mathcal{P}= \{P_1=\{u\in \mathbb{F}_q^{15} : 0\le \mathrm{wt}(u) \le 4 \},\ 
P_2=\{u\in \mathbb{F}_q^{15} : 5\le \mathrm{wt}(u) \le 9 \},\\
P_3=\{u\in \mathbb{F}_q^{15} : 10\le \mathrm{wt}(u) \le 14 \},\ 
P_4=\{u\in \mathbb{F}_q^{15} : \mathrm{wt}(u)=15\}\}.
\end{multline*}

This partition arises as the domain partition of the Hamming weight distribution
function $\Delta_5$ on $\mathbb{F}_q^{15}$. At the same time, the same partition
$\mathcal{P}$ is induced by other locally $(2t,2)$-bounded functions. For instance, consider the function
$f:\mathbb{F}_q^{15}\to \mathbb{F}_2^3$ defined by
$$
f(u)=\big(\mathbf{1}\{\mathrm{wt}(u)\ge 5\},\ 
\mathbf{1}\{\mathrm{wt}(u)\ge 10\},\ 
\mathbf{1}\{\mathrm{wt}(u)\ge 15\}\big),
$$
where the indicator function $\mathbf{1}\{\mathcal{E}\}$ equals $1$ if
$\mathcal{E}$ is true and $0$ otherwise. The function $f$ is also a locally $(2t,2)$-bounded function and induces the same domain partition $\mathcal{P}$.

Since $\mathcal{P}$ is locally $(2t,2)$-bounded for $t=1$, every Hamming ball of
radius $2$ intersects at most two blocks of $\mathcal{P}$. Consequently,
$\mathcal{P}$ admits a $(\mathcal{P},t)$-encoding with optimal redundancy for locally $(2t,2)$-bounded
partitions, as given in Construction~\ref{cons:Pt-encoding}.

From the transmitter’s perspective, the functions $\Delta_5$ and $f$ are
indistinguishable when only the partition $\mathcal{P}$ and the information
that it is locally $(2t,2)$-bounded are revealed. Consequently, the encoding
given by Construction~\ref{cons:Pt-encoding} applies to both functions, and the
decoding remains correct for the receiver.
\end{example}

The privacy provided by FCPCs is function class rather than
cryptographic in nature. It does not aim to protect the function completely, but instead limits the
information revealed about the exact function being computed. Most importantly, this
privacy does not interfere with the characterization of optimal redundancy or
the existence of optimal codes, since these depend only on the domain partition.

Function class privacy therefore emerges as a natural byproduct of the
partition-based encoding, highlighting an additional conceptual
distinction between FCCs and FCPCs.

\section{Conclusion}
\label{conclusion}

In this work, we introduced function-correcting partition codes as a natural generalization of function-correcting codes by shifting the focus from the function values to the domain partition induced by the function. By focusing on the domain partition rather than the function values, we showed that every FCC is a special case of an FCPC and that a single $(\mathcal{P},t)$-encoding can protect several functions simultaneously through the join of their domain partitions. This naturally led to the notion of partition gains, quantifying the savings in terms of redundancy and rate, compared to straightforward solutions. We established general redundancy bounds for such codes. Then we focus on how one code can be used to protect multiple linear functions through the intersection of their kernels. We also provide explicit FCPC constructions for locally bounded partitions and grouped weight partitions. Further, we introduce the partition graph to characterize when the problem size to obtain optimal redundancy can be reduced. We also discuss the partial privacy offered by this partition based encoding approach. Overall, the partition-based viewpoint offers a flexible and efficient approach to multi-function error protection.

\section{Acknowledgement}
This work was supported by a joint project grant to Aalto University and Chalmers University of Technology (PIs A. Graell i Amat and C. Hollanti) from the Wallenberg AI, Autonomous Systems and Software Program, and additionally by the Science and Engineering Research Board (SERB) of the Department of Science and Technology (DST), Government of India, through the J.C. Bose National Fellowship to Prof. B. Sundar Rajan.


\begin{thebibliography}{9}

\bibitem{LBWY2023}
A.~Lenz, R.~Bitar, A.~Wachter-Zeh, and E.~Yaakobi, ``Function-correcting codes,'' {\em IEEE Transactions on Information Theory}, vol.~69, no.~9,
  pp.~5604--5618, Sept.~2023.

\bibitem{XLC2024}
Q.~Xia, H.~Liu, and B.~Chen, ``Function-correcting codes for symbol-pair read
  channels,'' {\em IEEE Transactions on Information Theory}, vol.~70, no.~11,
  pp.~7807--7819, Nov.~2024.


  \bibitem{SSY2025}
A.~Singh, A.~K.~Singh, and E.~Yaakobi, ``Function-correcting codes for $b$-symbol
read channels,'' {\em arXiv preprint arXiv:2503.12894}, 2025.



\bibitem{PR2025}
R.~Premlal and B.~S.~Rajan, ``On function-correcting codes,'' {\em IEEE Transactions on Information Theory}, vol.~71, no.~8, pp.~5884--5897,
  Aug.~2025.


  \bibitem{GXZZ2025}
Y.~Zhang, Z.~Xu, X.~Zhang, and G.~Ge, ``Optimal redundancy of function-correcting
  codes,'' {\em IEEE Transactions on Information Theory}, vol.~71, no.~12,
  pp.~9458--9467, Dec.~2025.
  

\bibitem{RRHH2025ITW25}
C.~Rajput, B.~S.~Rajan, R.~Freij-Hollanti, and C.~Hollanti,
``Function-correcting codes for locally bounded functions,''
in Proc. {\em 2025 IEEE Information Theory Workshop (ITW)}, Sydney, Australia,
2025, pp.~851--856.


\bibitem{VSS2025}
G.~K.~Verma, A.~Singh, and A.~K.~Singh, ``Function-correcting $b$-symbol codes for
locally $(\lambda,\rho,b)$-functions,'' {\em IEEE Transactions on Information
Theory}, vol.~72, no.~1, pp.~331--341, Jan.~2026.


\bibitem{LS2025}
H.~Ly and E.~Soljanin,
``On the redundancy of function-correcting codes over finite fields,''
in Proc. {\em 2025 13th International Symposium on Topics in Coding (ISTC)},
Los Angeles, CA, USA, 2025, pp.~1--5.



\bibitem{VS2025}
G.~K.~Verma and A.~K.~Singh, ``On function-correcting codes in the Lee metric,''
{\em arXiv preprint arXiv:2507.17654v2}, 2025. 

\bibitem{HUR2025}
H.~K.~Hareesh, R.~Ummer N.~T., and B.~S.~Rajan, ``Plotkin-like bound and explicit
function-correcting code constructions for Lee metric channels,'' {\em arXiv preprint arXiv:2508.01702v3}, 2025. 

\bibitem{LL2025}
H.~Liu and H.~Liu, ``Function-correcting codes with homogeneous distance,''
{\em arXiv preprint arXiv:2507.03332}, 2025. 

\bibitem{SR2025}
S.~Sampath~H. and B.~S.~Rajan, ``On Plotkin bound for function-correcting codes for
$b$-symbol read channels,'' in Proc. {\em 2025 IEEE Information Theory Workshop
(ITW)}, Sydney, Australia, 2025, pp.~698--703.


\bibitem{RRHH2025}
C.~Rajput, B.~S.~Rajan, R.~Freij-Hollanti, and C.~Hollanti,
``Function-correcting codes with data protection,'' {\em arXiv preprint
arXiv:2511.18420}, 2025.








\bibitem{P1960}
M.~Plotkin, ``Binary codes with specified minimum distance,'' {\em IRE
Transactions on Information Theory}, vol.~6, no.~4, pp.~445--450, Oct.~1960.

\bibitem{DF2004}
D.~S.~Dummit and R.~M.~Foote, {\em Abstract Algebra}, 3rd~ed.
\newblock Hoboken, NJ, USA: Wiley, 2004.

\bibitem{LN}
D.~Quint, ``Lecture~3: Common knowledge and agreeing to disagree,''
lecture notes, Univ. of Wisconsin--Madison, Madison, WI, USA, 2014. [Online]. Available: \url{https://users.ssc.wisc.edu/~dquint/econ698/lecture%203.pdf}

\bibitem{BM2008}
J.~A.~Bondy and U.~S.~R.~Murty, {\em Graph Theory}.
\newblock New York, NY, USA: Springer, 2008.


\bibitem{CodeTable}
M.~Grassl, ``Bounds on the minimum distance of linear codes and quantum codes,''
available at: \url{https://www.codetables.de}, accessed in 2025.

\bibitem{M2023}
T.~M\"utze, ``Combinatorial Gray codes - an updated survey,'' {\em Electronic Journal of Combinatorics}, vol.~30, no.~DS26, pp.~1--99, 2023.

\bibitem{S1997}
C.~Savage, ``A survey of combinatorial Gray codes,'' {\em SIAM Review},
vol.~39, no.~4, pp.~605--629, 1997.

\bibitem{S2004}
K.~J.~Sankar, V.~M.~Pandharipande, and P.~S.~Moharir,
``Generalized Gray codes,'' in Proc. {\em 2004 International Symposium on Intelligent Signal Processing and Communication Systems (ISPACS)}, Seoul, Korea (South), 2004,
pp.~654--659.






\end{thebibliography}
\end{document}